\documentclass[11pt]{article}
\usepackage{amssymb,amsthm,mathtools,enumitem,graphicx,physics,xcolor,thmtools,authblk,bm,comment}
\usepackage[margin=1in]{geometry}
\usepackage{hyperref}

\newcommand{\R}{\mathbb R}
\newcommand{\C}{\mathbb C}

\newcommand{\bl}{\bm{\lambda}}
\newcommand{\pwcc}{\Delta_{\phi}}

\newcommand{\br}{\bm{\rho}}
\newcommand{\bx}{\vb{x}}
\newcommand{\by}{\vb{y}}
\newcommand{\bz}{\vb{z}}
\newcommand{\tl}{\tilde{\lambda}}
\newcommand{\tz}{\tilde{z}}
\newcommand{\aconst}{A_{\phi}}
\newcommand{\pconst}{P_{\phi}}
\newcommand{\sconst}{B}
\newcommand{\lsconst}{C}
\newcommand{\hcconst}{R}

\newcommand{\set}[1]{\left\{#1\right\}}

\usepackage{dsfont}
\newcommand{\indicator}[1]{\mathds{1}_{#1}}

\allowdisplaybreaks

\numberwithin{equation}{section}

\theoremstyle{plain}
\newtheorem{theorem}{Theorem}[section]
\newtheorem{lemma}[theorem]{Lemma}
\newtheorem{proposition}[theorem]{Proposition}
\newtheorem{corollary}[theorem]{Corollary}

\theoremstyle{definition}
\newtheorem{definition}[theorem]{Definition}

\theoremstyle{remark}
\newtheorem{claim}[theorem]{Claim}

\newtheorem{remark}[theorem]{Remark}

\begin{document}

\title{Analyticity for locally stable hard-core gases via recursion}
\author{Qidong He}
\affil{Department of Mathematics, Rutgers University}
\date{}

\maketitle

\abstract

In their recent works [Comm. Math. Phys. \textbf{399}:1 (2023)] and [arXiv:2109.01094], Michelen and Perkins proved that the pressure of a system of particles with repulsive pair interactions is analytic for activities up to $e\Delta_{\phi}(\beta)^{-1}$, where $\Delta_{\phi}(\beta)\in(0,C_{\phi}(\beta)]$ is a constant they called the potential-weighted connective constant. This paper extends their method to locally stable, tempered, and hard-core pair potentials. Our main result is that the pressure of such a system is analytic for activities up to $e^{2-2W(eA_{\phi}(\beta)/\Delta_{\phi}(\beta))}\Delta_{\phi}(\beta)^{-1}e^{-(\beta C+1)}$, where $C\ge0$ is the local stability constant, $W(\cdot)$ the Lambert $W$-function, $A_{\phi}(\beta)$ the contribution from the attraction in the pair potential to the temperedness constant, and $\Delta_{\phi}(\beta)\in[A_{\phi}(\beta),C_{\phi}(\beta)]$ a counterpart of the constant defined by Michelen and Perkins. The main ingredients in the proof include a recursive identity for the one-point density tailored to locally stable hard-core potentials and a corresponding notion of modulations of an activity function. In the high-temperature regime, our result surpasses the classical Penrose-Ruelle bound of $C_{\phi}(\beta)^{-1}e^{-(\beta C+1)}$ by at least a factor of $e^{2}$.

\tableofcontents

\section{Introduction}

One of the central objectives of equilibrium statistical mechanics is understanding the nature of phase transitions.
While phase transitions are abundantly observed and utilized in the real world, producing a rigorous mathematical proof of their existence in a continuum model with realistic interactions has largely remained an open problem \cite{lebowitz1999liquid,lewin2015crystallization,widom1970new}.
Here, we focus on the problem of proving the \emph{absence} of phase transitions.
We do this within the framework of the Lee-Yang theory, which associates non-analytic points of the infinite-volume thermodynamic potentials of a system to the onset of phase transitions \cite{yang1952statistical}.
A multitude of existing results in this direction are based on the series expansion method, which proves analyticity (hence the absence of phase transitions) by explicitly writing the thermodynamic potential of interest as a power series and demonstrating its convergence \cite{fernandez2007analyticity,mayer1941molecular,penrose1963convergence,procacci2017convergence,ruelle1963correlation}.
However, the region of convergence for these series often encompasses more than is physically necessary and can therefore face non-physical obstructions, e.g., singularities on the negative real-axis for activity expansions of the pressure in repulsive particle models (whereas only the positive activities are physically meaningful) \cite{achlioptas2021lov,groeneveld1962two}.
To avoid the non-physical singularities, we adopt the analytic method developed by Michelen and Perkins  \cite{michelen2021potential,michelen2023analyticity} that allows one to work directly in a thin complex neighborhood of the positive real axis.
We show that, for systems of hard-core particles with locally stable pair interactions, this method leads to a temperature- and interaction-dependent improvement over the classical bound on the radius of convergence of the series expansion computed by Ruelle \cite{ruelle1963correlation} and Penrose \cite{penrose1963convergence}.
In particular, the improvement is most pronounced in the high-temperature regime and apparently matches, in the infinite-temperature limit, the state-of-the-art bounds for pure hard-core systems due to Michelen and Perkins \cite{michelen2021potential,michelen2023analyticity}.

\subsection{The model}
\label{sec:introduction_model}

In this paper, we consider a system of classical particles in $\R^{d}$, interacting via a radially symmetric, translation-invariant pair potential $\phi:\R^{d}\rightarrow\R\cup\set{\infty}$.
We will assume that $\phi$ is \emph{locally stable}, \emph{tempered}, and has a hard core:
\begin{enumerate}
\item (local stability) There exists a constant $\lsconst\ge0$ such that, for all $y\in\R^{d}$, $n\ge 1$, and sets of particles $\bm{x}=\set{x_{1},\dots,x_{n}}\subset\R^{d}$ with finite potential energy 
\begin{equation}
\label{eqn:introduction_potential-energy}
U(\bm{x}):=\sum_{i<j}\phi(x_{i}-x_{j})<\infty,
\end{equation} 
the potential energy between $y$ and $\bm{x}$ is uniformly bounded below:
\begin{equation}
\label{eqn:introduction_local-stability}
\sum_{i=1}^{n}\phi(x_{i}-y)\ge -\lsconst.
\end{equation}
\item (temperedness) The pair potential decays over long distances in the sense that
\begin{equation}
\label{eqn:introduction_temperedness}
C_{\phi}(\beta):=\int_{\R^{d}}\dd{w}\abs{1-e^{-\beta\phi(w)}}<\infty,
\end{equation}
where $\beta$ is the inverse temperature.
The property \eqref{eqn:introduction_temperedness} is originally termed ``regularity'' in \cite[\S4.1]{ruelle1969statistical}; here, we follow the naming convention in \cite{michelen2023analyticity}.
\item (hard core) There exists $\hcconst>0$ such that $\phi(w)=\infty$ for $\abs{w}<\hcconst$ and is finite otherwise.
\end{enumerate}

Strictly speaking, the hard-core assumption is redundant because a radially symmetric pair potential with an attractive part cannot be locally stable \emph{unless} it has a hard core.
Nevertheless, it is convenient to make the presence of the hard core explicit for the sake of computation in this paper.

We comment that it is not possible to deduce that a (radially symmetric) pair potential $\phi$ is locally stable \emph{backwards} from the temperedness and hard-core assumptions.
The reason is that the integrability criterion \eqref{eqn:introduction_temperedness} for temperedness that we use here does not guarantee \emph{pointwise} long-range decay of the potential, whereas local stability \eqref{eqn:introduction_local-stability} is intrinsically a pointwise property.
For instance, we can easily obtain a counterexample by modifying any tempered, hard-core potential $\phi$ so that $\phi=-1$ on the concentric spherical shells $\partial B_{nR}(0)$, where $n\ge 1$.

Finally, we note that the setting we work in is not the most general one in which the statistical mechanical theory of particles has been historically formulated and studied.
As was established in the classic reference by Ruelle \cite{ruelle1969statistical}, for a classical ensemble to yield physically meaningful predictions, it suffices to require that the interaction be stable and tempered: a pair potential $\phi$ is \emph{stable} if the potential energy of every finite set of particles is bounded below by a (fixed) constant multiple of the number of particles.
Our assumption of local stability is easily seen to be a stronger requirement than mere stability: assuming without loss of generality that the set $\bm{x}=\set{x_{1},\dots,x_{n}}$ has finite potential energy, we get, as in \cite[\S3.2.5]{ruelle1969statistical}, that
\begin{equation}
\label{eqn:introduction_stability}
U(\bm{x})=\frac{1}{2}\sum_{i=1}^{n}\sum_{j\ne i}\phi(x_{j}-x_{i})\ge -\frac{\lsconst}{2}n.
\end{equation}
Thus, an interesting direction for future work would be to replace our assumption of local stability with that of stability, although we conjecture that this may be especially challenging to achieve (see \S\ref{sec:introduction_model_local-stability}).

\subsection{Main result}

From here on, we absorb the inverse temperature $\beta$ into the pair potential $\phi$ and suppress the former from the notation.
For a bounded, Lebesgue measurable region $\Lambda\subset\R^{d}$, the grand-canonical partition function at activity $\lambda\ge 0$ is 
\begin{equation}
\label{eqn:introduction_partition-function}
Z_{\Lambda}(\lambda):=1+\sum_{n=1}^{\infty}\frac{\lambda^{n}}{n!}\int_{\Lambda^{n}}\dd{x}_{1}\dots\dd{x}_{n}e^{-U(x_{1},\dots,x_{n})}.
\end{equation}
The infinite-volume pressure of the system is given by
\begin{equation}
p_{\phi}(\lambda):=\lim_{\Lambda\uparrow\R^{d}}\frac{1}{\abs{\Lambda}}\log Z_{\Lambda}(\lambda),
\end{equation}
where the convergence $\Lambda\uparrow\R^{d}$ is in the sense of Fisher \cite{fisher1964free}; a proof of the existence of this limit can be found in \cite[\S3]{ruelle1969statistical}.
The main result of this paper is as follows.

\begin{theorem}
\label{thm:analyticity}
Let $\phi$ be a translation-invariant, radially symmetric, locally stable, and tempered pair potential (with a hard core).
Denote by
\begin{equation}
\label{eqn:attraction-constant}
\aconst:=\int_{\set{w\in\R^{d}\mid\phi(w)<0}}\dd{w}(e^{-\phi(w)}-1)
\end{equation}
the contribution from the attraction in the pair potential to the temperedness constant $C_{\phi}$ (cf. \eqref{eqn:introduction_temperedness}), and let $\pwcc$ be as defined in \eqref{def:tree-recursions_pwcc} (for now, it suffices to note that $\pwcc$ is a $\phi$-dependent constant in the interval $[\aconst,C_{\phi}]$).
Then, the infinite-volume pressure $p_{\phi}(\lambda)$ is analytic for 
\begin{equation}
\lambda\in\left[0,e^{2-2W(e\aconst/\pwcc)}\pwcc^{-1}e^{-(\lsconst+1)}\right),
\end{equation}
where $W(\cdot)$ is the Lambert $W$-function.
\end{theorem}

\begin{remark}
The ingredients needed for computing $\pwcc$ include 
\begin{enumerate}
\item the definition \eqref{eqn:tree-recursions_good-sequences} of ``good sequences'' in $\R^{d}$ of finite length, 
\item the construction of the reference modulating function $\vec{\gamma}_{c}$ in Definition \ref{def:tree-recursions_canonical-modulating-function}
\item the expression of $V_{k,\vb{0}}(\vec{\gamma}_{c})$ for $k\ge 1$ using \eqref{cor:contraction-properties_real-contraction_Vk}, and 
\item the definition of $\pwcc$ as an infimum over the suitable $V_{k,\vb{0}}(\vec{\gamma}_{c})$'s in Definition \ref{def:tree-recursions_pwcc}.
\end{enumerate}
We will show in Lemma \ref{lem:contraction-properties_real-contraction_pwcc-upper-bound} that $V_{1,\vb{0}}(\vec{\gamma}_{c})=C_{\phi}$, so it is always the case that $\pwcc\le C_{\phi}$.
Determining $V_{k,\vb{0}}(\vec{\gamma}_{c})$ for interesting locally stable potentials and $k\ge 2$ seems to be a daunting task generally inaccessible to manual computation.
To make matters worse, in contrast to the purely repulsive case \cite[Lemma 9]{michelen2023analyticity}, we have no theoretical guarantee at this moment that computing $V_{k,\vb{0}}(\vec{\gamma}_{c})$ for larger $k$'s will lead to arbitrarily tight bounds on $\pwcc$.
However, computing $V_{2,\vb{0}}(\vec{\gamma}_{c})$ for the hard-rod model in 1D with an attractive long-range Kac interaction shows that $\pwcc<C_{\phi}$ for all inverse temperature $\beta$ in that case.
We leave a more in-depth study of the constant $\pwcc$ as a possible direction for future work.
\end{remark}

\subsubsection{Comparison with existing bounds}
\label{sec:introduction_model_comparison}

Before discussing the proof strategy, let us briefly compare our bound to several existing ones in the literature.
Since the dependence of $\pwcc$ on the potential $\phi$ and the inverse temperature $\beta$ (which is absorbed into $\phi$) is complicated in general (see Definition \ref{def:tree-recursions_pwcc}), we will only consider below the worst-case performance of our bound, i.e., when $\pwcc=C_{\phi}$, for the sake of argument.

We consider first the general class of stable, tempered pair potentials.
A classical result of Ruelle \cite{ruelle1963correlation} and Penrose \cite{penrose1963convergence}, obtained from a study of the Kirkwood-Salzburg equations, states that the activity expansion of the pressure converges in the disk 
\begin{equation}
\label{eqn:introduction_Penrose-Ruelle}
\set{\lambda\in\C\mid\abs{\lambda}\le C_{\phi}^{-1}e^{-(2\sconst+1)}},
\end{equation} 
where $\sconst\ge 0$ is the stability constant.
In the case that stability is replaced by local stability, taking $\sconst=\lsconst/2$ as in \eqref{eqn:introduction_stability}, we see that our result outperforms the Penrose-Ruelle bound by a factor of (at least) $e^{2-2W(e\aconst/C_{\phi})}\in(1,e^{2}]$.

The state-of-the-art lower bound on the radius of convergence of the activity expansion appears to be due to Procacci and Yuhjtman \cite{procacci2017convergence}: through careful manipulations of the graphs that arise in the expansion, they prove convergence in the disk
\begin{equation}
\set{\lambda\in\C\mid\abs{\lambda}\le \hat{C}_{\phi}^{-1}e^{-(\sconst+1)}},
\end{equation} 
where the so-called \emph{weak temperedness contant}
\begin{equation}
\hat{C}_{\phi}:=\int_{\R^{d}}\dd{w}(1-e^{-\abs{\phi(w)}})
\end{equation} 
replaces the temperedness constant $C_{\phi}$ in the Penrose-Ruelle bound \eqref{eqn:introduction_Penrose-Ruelle}.
Their improvement over the latter is most prominent in the low-temperature regime due to the factor of $e^{\sconst}$, which they comment is thanks to their ability to ``use in an optimal way the stability condition of the pair potential.''
They also remark that their improvement is ``much less sensitive at high temperatures,'' in which case $e^{\sconst}\approx1$ and $\hat{C}_{\phi}\approx C_{\phi}$, so they recover the Penrose-Ruelle bound.
In comparison, our gain factor of (at least) $e^{2-2W(e\aconst/C_{\phi})}$ over the latter is close to $e^{2}$ in the high-temperature regime ($\beta\ll 1$), where the hard-core contribution dominates the temperedness constant $C_{\phi}$, so our bound outperforms theirs at high temperatures.
We should also note that the more intricate temperature dependence of $\pwcc$ as compared to $C_{\phi}$ can lower the temperature threshold for our bound to be an improvement over theirs, but this will depend strongly on the detail of the potential $\phi$.

The factor of $e^{-\lsconst}$ in our result originates from Lemma \ref{lem:one-point-densities_uniform-boundedness}, where we use local stability to prove the uniform boundedness of an infinite family of activity functions called \emph{modulations} (of a given activity function).
As we will explain in \S\ref{sec:introduction_model_local-stability}, this is an essential step in the argument with no obvious way around.
Thus, the improvement factor of $e^{\sconst}$ in \cite{procacci2017convergence} seems out of reach without significant new insight.

For the subclass of repulsive (and hard-core) pair potentials, $\lsconst=0$ and $\aconst=0$, and our result shows analyticity for positive activities up to (at least) $eC_{\phi}^{-1}$, coinciding with that obtained by Michelen and Perkins in \cite{michelen2021potential}. 
We refer the reader to the discussions in \cite{michelen2021potential,michelen2023analyticity} for comparisons of this bound to other existing results in the literature for repulsive potentials.
As the present paper is based on the method that they developed, readers who are familiar with their works might wonder why the bound we obtain in this paper for locally stable systems has an uncanny temperature dependence (i.e., the factor of $e^{2-2W(e\aconst/\pwcc)}$).
We will argue in \S\ref{sec:zero-freeness_origin} that this dependence stems from a trade-off between the desired contraction property and a necessary self-map condition in analyzing the recursion that lies at the heart of the method.
In short, for systems with attractive interactions, the self-map condition eats away the range of activities where the contraction property holds, an effect that becomes dramatically more severe at lower temperatures.
On the other hand, this adverse interplay is largely invisible for purely repulsive systems, hence Michelen and Perkins' temperature-independent gain factor of $e^{2}$ over the classical bound of Groeneveld \cite{groeneveld1962two}.

Finally, we note that the method of Michelen and Perkins \cite{michelen2021potential,michelen2023analyticity}, on which the present paper is based, has roots in theoretical computer science \cite{weitz2006counting}, and so our result may have algorithmic implications that could be interesting to explore \cite{jenssen2024quasipolynomial,michelen2022strong,sinclair2017spatial}.
We refer the reader to \cite[\S1.4]{michelen2023analyticity} and the references therein for details.

\subsubsection{Overview of strategy}
\label{sec:introduction_model_strategy}

Our strategy for proving Theorem \ref{thm:analyticity} combines ideas from both the original paper of Michelen and Perkins \cite{michelen2021potential} and their follow-up work \cite{michelen2023analyticity}.
We provide here an outline of the strategy, which readers familiar with the works of Michelen and Perkins on purely repulsive systems might recognize.
We also indicate where new ideas are introduced in order to bring the strategy to fruition in the case of locally stable interactions.

To prove the analyticity of the infinite-volume pressure, we will work with complex-valued activity functions $\bl$ rather than the real, constant activity parameter $\lambda$ in terms of which the partition function \eqref{eqn:introduction_partition-function} is presently defined.
Our starting point is an expression (see Lemma \ref{lem:one-point-densities_log-Z}) of the log of the partition function $Z_{\Lambda}(\bl)$ in terms of an integral of certain one-point densities.
To evaluate these one-point densities, we will use an integral expression (see Theorem \ref{thm:recursion}) of the one-point density of any activity function $\bl$ in terms of the densities of other activity functions derived from $\bl$.
This enables us to compute the one-point densities that appear in the identity for $\log Z_{\Lambda}(\bl)$ \emph{recursively}: computing the initial one-point densities requires computing another set of one-point densities, which requires computing yet another set of one-point densities, and so forth.
In this way, the recursion generates an infinite hierarchy of activity functions which is up for further analysis.

The ideas so far are based on \cite[\S2]{michelen2021potential}.
However, as we will explain in \S\ref{sec:introduction_model_local-stability}, we already have to make an important modification to the recursive identity proven in \cite[Theorem 8]{michelen2021potential}.
In the language of \cite[\S4.1]{michelen2023analyticity}, the damping function that our identity generates depends explicitly on the distance between the root of a tree recursion and its children: it only takes a similar form to that in \cite{michelen2021potential} if this distance exceeds the hard-core radius $\hcconst$ (see \eqref{eqn:one-point-densities_modulation}).
This is an essential refinement of the treatment in \cite{michelen2023analyticity} that allows us to take advantage of local stability later on.

To analyze the hierarchy of activity functions, we will follow the strategy in the follow-up paper \cite{michelen2023analyticity} of Michelen and Perkins.
We note that the analysis in both their papers is based on showing that the recursive identity generating the hierarchy satisfies a contraction property at low activities.
The difference is that whereas their original paper considers the identity as is, their follow-up paper allows for the possibility that some (high) iteration of it is a contraction even though one iteration might not be.
On a high level, this is similar to a well-known extension of the Banach contraction principle due to Bryant \cite{bryant1968remark}, who noted that the usual conclusions of the Banach contraction principle still hold for an operator $T$ if only a certain power of $T$ is a contraction.

To characterize high iterations of our recursive identity, we will adopt the formalism of finite-depth tree recursions introduced in \cite[\S4.1]{michelen2023analyticity}.
This requires an original construction of the \emph{damping functions} in the tree recursions because our recursion differs from theirs.
There is also a more serious question regarding to which family of activity functions we should apply the tree recursions.
For the later arguments to work, this family must be uniformly bounded, include all the interpolating activity functions used to prove the key identities (see the assumption of nonzero partition functions in Lemma \ref{lem:one-point-densities_log-Z} and Theorem \ref{thm:recursion}), and be closed under recursion. 
The family generated directly by the recursion from a single activity function is inadequate because it violates the second requirement of including the interpolating activity functions.
In other words, one must expand the above hierarchy of activity functions in order to proceed further.
In \cite{michelen2021potential,michelen2023analyticity}, Michelen and Perkins opted to work with the family of \emph{pointwise contractions} of a given activity function.
As each iteration of the recursion amounts to multiplication by a suitable damping function, which is a pointwise contraction, the product of a damping function and a pointwise contraction remains a pointwise contraction, which thus solves their problem.
Unfortunately, this will not work with attractive pair potentials because the ``damping functions'' that correspond to the recursion are no longer pointwise contractions, so arbitrary products of ``damping functions'' will generally be unbounded.
Instead, we will define the family of \emph{modulations} of a given activity function (see Definition \ref{def:tree-recursions_modulation}), which includes all the interpolating activity functions mentioned before, and tailor the ``damping functions,'' which we call \emph{modulating functions} (see Theorem \ref{thm:tree-recursions_density-correspondence}), such that the product of a modulating function and a modulation remains a modulation (see Claim \ref{clm:tree-recursions_density-correspondence_interpolation} and Lemma \ref{lem:one-point-densities_uniform-boundedness}).
In particular, we will construct the modulations to be uniformly bounded by local stability, which will enable us to study this larger hierarchy of activity functions using tree recursions in the same way as in \cite{michelen2023analyticity}.

The analysis becomes genuinely interesting when we investigate the contraction properties of tree recursions for locally stable potentials.
In the corresponding analysis in \cite{michelen2023analyticity} for purely repulsive systems, a quantity $\pwcc$ they call the \emph{potential-weighted connective constant} naturally arises that ``captures the interplay between the potential and the geometry of the space.''
Remarkably, the quantity $\pwcc$ is defined there with a submultiplicative sequence $V_{k}(\vec{\gamma})$ \cite[Lemma 9]{michelen2023analyticity} using Fekete's lemma, in direct parallel to the definition of the connective constant for lattices \cite{duminil2012connective}.
We will show that while a similar sequence $V_{k}(\vec{\gamma})$ appears in our contraction condition, the submultiplicativity property breaks down at the first whiff of attraction.
Hence, Fekete's lemma does not apply, and there is, sadly, no longer an analogy to the connective constants of lattices.

Another difficulty largely hidden in the analysis of purely repulsive systems is the necessity of introducing a self-map condition to ensure that the same contraction property holds at all places in the hierarchy.
Striking a balance between the self-map condition and the contraction property ultimately leads to the sharp temperature dependence of the activity threshold we compute for locally stable potentials, in contrast to the threshold computed by Michelen and Perkins for repulsive systems which does not feature this dependence.
We have touched upon this briefly in \S\ref{sec:introduction_model_comparison} and will return to it in \S\ref{sec:zero-freeness_origin}.

Once the above technicalities are settled, we will show by interpolation from zero activity that the finite-volume pressures are uniformly bounded for all activities at which both the self-map condition and the (iterated) contraction property hold.
By Vitali's convergence theorem, the analyticity of the infinite-volume pressure follows.

\subsubsection{Role of local stability}
\label{sec:introduction_model_local-stability}

Here, we comment on the role of local stability in this paper.

Throughout the paper, local stability appears most explicitly in Lemma \ref{lem:one-point-densities_uniform-boundedness}, where we use it to deduce that an infinite family of \emph{modulating functions} are uniformly bounded.
This is significant for the following reasons.
As we have mentioned in \S\ref{sec:introduction_model_strategy}, a large part of the paper is devoted to the study of an infinite hierarchy of activity functions \emph{modulated} by these modulating functions.
To this end, the uniform boundedness property allows us to accomplish two things.
First, it allows us to use the uniform continuity of the one-point density (see Lemma \ref{lem:zero-freeness_Z-uniform-continuity}) to control the effect of rescaling the activity function at the top of the hierarchy on the one-point densities of \emph{all} the activity functions located below it.
This strong control is key for the interpolation step in the proof of Theorem \ref{thm:zero-freeness}, where we interpolate from zero activity where the one-point density is known explicitly.
Second, it allows us to apply the contraction properties that we derive (see Corollary \ref{cor:contraction-properties_complex-contraction_final-result}), which holds only for activity functions taking values in a bounded neighborhood in $\C$, to arbitrary activity functions in the hierarchy, provided that the activity function at the top is sufficiently constrained.
This ultimately produces the desired contradiction in the proof of Theorem \ref{thm:zero-freeness}, thus allowing us to deduce analyticity until the contraction property breaks down.

It is important to note that local stability is not as straightforward to use in practice as the above discussion might suggest.
Specifically, applying it requires the ambient configuration $\vb{x}$ to have finite energy (see \eqref{eqn:introduction_potential-energy}), but n\"aively applying the original recursive identity of Michelen and Perkins in this context easily generates ambient configurations that violate the hard-core constraint.
The way out that we propose is to take the complementary hard-core interaction (see \S\ref{sec:introduction_model}) into explicit account by showing that, in all the relevant cases, the placement of one particle during the recursion forbids future addition of particles to its hard-core excluded volume.
Ultimately, this results in a different recursive relation (see Theorem \ref{thm:recursion}) to analyze and a markedly more involved construction of modulating functions in \S\ref{sec:tree-recursions}.

It is noteworthy that our work along with those by Michelen and Perkins \cite{michelen2021potential,michelen2023analyticity} exhaust the set of Gibbs point processes with pair interactions that have bounded Papangelou intensities (which, informally, means that the gain in energy from adding a particle to any given configuration is uniformly bounded from below).
With such a Gibbs point process, a criterion due to Georgii and K\"uneth \cite{georgii1997stochastic} implies stochastic domination by an appropriate Poisson point process (see \cite[Example 2.2]{georgii1997stochastic}).
Poisson domination has played a crucial role in \cite{michelen2021potential,michelen2023analyticity} for proving the uniqueness of Gibbs measures up to the same activity threshold as they have established the pressure to be analytic.
Although we do not require it in the present paper, as we are focused on the analyticity of the pressure (though proving the uniqueness of Gibbs measures for locally stable models up to the same analyticity threshold could be a direction for future work), the modulating functions that we have to work with (see \eqref{eqn:tree-recursions_canonical-modulating-function-good-sequence}) are still easily seen to resemble the Papangelou intensities.
Therefore, the boundedness requirement of the Papangelou intensity appears to be the most significant barrier to extending the contraction method to more general particle systems like those with stable, tempered pair interactions, a program envisioned in \cite{anand2023perfect}.

\bigskip

The rest of the paper is organized as follows.
In \S\ref{sec:one-point-densities}, we introduce complex-valued activity functions and (generalized) one-point densities, and prove two key identities that the one-point densities satisfy.
In \S\ref{sec:tree-recursions}, we import the formalism of tree recursions from \cite{michelen2023analyticity} and construct a version tailored to the recursive identity for locally stable hard-core models.
In \S\ref{sec:contraction-properties}, we study the contraction properties of tree recursions.
Finally, in \S\ref{sec:zero-freeness}, we combine the recursive structure and the contraction properties to prove Theorem \ref{thm:analyticity}.

\section{One-point densities}
\label{sec:one-point-densities}

In this section, we study complex-valued activity functions and their one-point densities.

For a bounded, Lebesgue measurable region $\Lambda\subset\R^{d}$, an \emph{activity function} is a \emph{bounded}, Lebesgue measurable function $\bl:\Lambda\rightarrow\C$.
Its associated grand-canonical partition function is
\begin{equation}
\label{eqn:one-point-densities_partition-function}
Z_{\Lambda}(\bl):=1+\sum_{n=1}^{\infty}\frac{1}{n!}\int_{\Lambda^{n}}\dd{x}_{1}\dots\dd{x}_{n}e^{-U(x_{1},\dots,x_{n})}\prod_{i=1}^{n}\bl(x_{i}).
\end{equation}
We also introduce the \emph{(generalized) one-point density} associated to the activity function $\bl$: if $Z_{\Lambda}(\bl)\ne 0$, we define, for each point $v\in\R^{d}$,
\begin{equation}
\label{eqn:one-point-densities_def}
\rho_{\Lambda,\bl}(v):=\bl(v)\frac{Z_{\Lambda}(\bl e^{-\phi(v-\cdot)})}{Z_{\Lambda}(\bl)}.
\end{equation}
The name \emph{density} alludes to the fact that for any real-valued activity function $\bl$ and Lebesgue measurable subset $A\subseteq\Lambda$, the integral
\begin{equation}
\int_{A}\dd{v}\rho_{\Lambda,\bl}(v)
\end{equation}
is the expected number of particles in $A$ at activity $\bl$.
The requirement that $Z_{\Lambda}(\bl)\ne 0$ in \eqref{eqn:one-point-densities_def} slightly complicates several proofs to follow which are based on the interpolation idea.
To remedy this issue, we will introduce in \S\ref{sec:tree-recursions} a notion of zero-freeness for an activity function that encompasses all the interpolating activity functions we will encounter.

From here on, we fix the region $\Lambda$ on which the activity functions are supported and suppress it from our notation, e.g., we will write $\rho_{\bl}$ instead of $\rho_{\Lambda,\bl}$.

We now present two identities involving the one-point densities.
The first identity, taken directly from \cite{michelen2023analyticity}, expresses the log of the partition function at activity $\bl$ as an integral of the one-point densities of other activity functions obtained from $\bl$.

\begin{lemma}[{\cite[Lemma 7]{michelen2023analyticity}}]
\label{lem:one-point-densities_log-Z}
Let $\bl:\Lambda\rightarrow\C$ be an activity function.
Suppose that $Z(\bl_{s})\ne 0$ for all $s\ge 0$, where $\bl_{s}(\cdot):=\indicator{B_{s}(0)^{c}}(\cdot)\bl(\cdot)$.
Then,
\begin{equation}
\label{eqn:one-point-densities_log-Z}
\log Z(\bl)=\int_{\Lambda}\dd{x}\rho_{\hat{\bl}_{x}}(x),
\end{equation}
where $\hat{\bl}_{x}$ denotes the activity function
\begin{equation}
\hat{\bl}_{x}(\cdot):=\indicator{B_{d(0,x)}(0)^{c}}(\cdot)\bl(\cdot).
\end{equation}
\end{lemma}

\begin{remark}
We prove Lemma \ref{lem:one-point-densities_log-Z} in detail below because the original proof of \cite[Lemma 7]{michelen2023analyticity} contains a flaw: the last line in their rewriting of $Z(\bl_{t})-1$ depends on the parameter $t$ in multiple places, so one cannot use the fundamental theorem of calculus right away to compute $\dv{t}Z(\bl_{t})$.
Nonetheless, our proof confirms that the conclusion of \cite[Lemma 7]{michelen2023analyticity} is correct.
\end{remark}

\begin{proof}
Since $Z(\bl_{0})=Z(\bl)$ and $Z(\bl_{\infty})=Z(0)=1$,
\begin{equation}
\log Z(\bl)=-\int_{0}^{\infty}\dd{s}\dv{s}\log Z(\bl_{s})=-\int_{0}^{\infty}\dd{s}\frac{\dv{s}Z(\bl_{s})}{Z(\bl_{s})}.
\end{equation}
It may be helpful to note that $\bl_{\infty}$ is only a notational convenience and does not cause any convergence issues with the improper integral: since $\bl$ is supported on a bounded region $\Lambda$, $\bl_{s}\equiv 0$ for all sufficiently large $s$, so we could have truncated the integral at some finite but large $s$.
To compute the derivative $\dv{s}Z(\bl_{s})$, we change to spherical coordinates as follows.
\begin{equation}
Z(\bl_{s})=1+\sum_{n=1}^{\infty}\frac{1}{n!}\int_{0}^{\infty}\dd{r}_{1}\dots\int_{0}^{\infty}\dd{r}_{n}\int_{\partial B_{r_{1}}(0)}\dd{x}_{1}\dots\int_{\partial B_{r_{n}}(0)}\dd{x}_{n} e^{-U(\bx)}\prod_{i=1}^{n}\bl_{s}(x_{i}).
\end{equation}
Notice that $\bl_{s}(x_{i})=\indicator{r_{i}\ge s}\bl(r_{i})$ for $x_{i}\in\partial B_{r_{i}}(0)$.
Thus,
\begin{equation}
Z(\bl_{s})=1+\sum_{n=1}^{\infty}\frac{1}{n!}\int_{s}^{\infty}\dd{r}_{1}\dots\int_{s}^{\infty}\dd{r}_{n}\int_{\partial B_{r_{1}}(0)}\dd{x}_{1}\dots\int_{\partial B_{r_{n}}(0)}\dd{x}_{n} e^{-U(\bx)}\prod_{i=1}^{n}\bl(x_{i}).
\end{equation}
Differentiating with respect to $s$, we get that
\begin{equation}
\begin{split}
\dv{s}Z(\bl_{s}){}
&=-\sum_{n=1}^{\infty}\frac{1}{n!}\cdot n\int_{\partial B_{s}(0)}\dd{w}\bl(w)\int_{(\R^{d})^{n-1}}\dd{\vb{x}} e^{-U(\bx,w)}\prod_{i=1}^{n-1}\bl_{s}(x_{i}) \\
{}&=-\int_{\partial B_{s}(0)}\dd{w}\bl(w)\sum_{n=0}^{\infty}\frac{1}{n!}\int_{(\R^{d})^{n}}\dd{\vb{x}} e^{-U(\bx)}\prod_{i=1}^{n}\left[\bl_{s}(x_{i})e^{-\phi(w-x_{i})}\right] \\
{}&=-\int_{\partial B_{s}(0)}\dd{w}\bl(w)Z(\bl_{s}e^{-\phi(w-\cdot)}).
\end{split}
\end{equation}
Therefore,
\begin{equation}
\log Z(\bl)
=\int_{0}^{\infty}\dd{s}\int_{\partial B_{s}(0)}\dd{w}\bl(w)\frac{Z(\bl_{s}e^{-\phi(w-\cdot)})}{Z(\bl_{s})}
=\int_{\R^{d}}\dd{w}\rho_{\hat{\bl}_{w}}(w)
\end{equation}
after noting that $\bl(w)=\hat{\bl}_{w}(w)=\bl_{s}(w)$ for $w\in\partial B_{s}(0)$.
\end{proof}

The second identity is an adaptation of \cite[Theorem 8]{michelen2023analyticity} to the case of locally stable, hard-core potentials and lies at the heart of the whole proof.
It expresses the one-point density of an activity function $\bl$ in terms of the one-point densities of other activity functions obtained from $\bl$.

\begin{theorem}
\label{thm:recursion}
Let $\bl:\Lambda\rightarrow\C$ be an activity function and $v\in\R^{d}$.
Suppose that $Z(\bl_{s})\ne 0$ for all $s\ge 0$, where $\bl_{s}(\cdot):=\left[1+(e^{-\phi(v-\cdot)}-1)\indicator{B_{s}(v)}(\cdot)\right]\bl(\cdot)$.
Then, 
\begin{equation}
\label{eqn:one-point-densities_recursion}
\rho_{\bl}(v)=\bl(v)\exp{-\int_{\R^{d}}\dd{w}\left[1-e^{-\phi(v-w)}\right]\rho_{\bl_{v\rightarrow w}}(w)},
\end{equation}
where $\bl_{v\rightarrow w}$ denotes the activity function
\begin{equation}
\label{eqn:one-point-densities_modulation}
\bl_{v\rightarrow w}(\cdot):=
\begin{cases}
\indicator{B_{d(v,w)}(v)^{c}}(\cdot)\bl(\cdot) & \text{if }d(v,w)<\hcconst \\
\indicator{B_{\hcconst}(v)^{c}}(\cdot)\left[1+(e^{-\phi(v-\cdot)}-1)\indicator{B_{d(v,w)}(v)}(\cdot)\right]\bl(\cdot) & \text{else}
\end{cases}.
\end{equation}
\end{theorem}

\begin{proof}
Note that $Z(\bl_{0})=Z(\bl)$ and $Z(\bl_{\infty})=Z(\bl e^{-\phi(v-\cdot)})$, so $\bl_{s}(\cdot)$ interpolates between $\bl(\cdot)$ and $\bl(\cdot)e^{-\phi(v-\cdot)}$ (where, as in the proof of Lemma \ref{lem:one-point-densities_log-Z}, $\bl_{\infty}$ is only a notational convenience):
\begin{equation}
\log Z(\bl e^{-\phi(v-\cdot)})-\log Z(\bl)
=\int_{0}^{\infty}\dd{s}\frac{\dv{s}Z(\bl_{s})}{Z(\bl_{s})}.
\end{equation}
To evaluate the derivative $\dv{s}Z(\bl_{s})$, we change to spherical coordinates centered around $v$ as follows.
\begin{equation}
\begin{multlined}
Z(\bl_{s})=1+\sum_{n=1}^{\infty}\frac{1}{n!}\int_{0}^{\infty}\dd{r}_{1}\dots\int_{0}^{\infty}\dd{r}_{n}\int_{\partial B_{r_{1}}(v)}\dd{x}_{1}\dots\int_{\partial B_{r_{n}}(v)}\dd{x}_{n} e^{-U(\bx)}\prod_{i=1}^{n}\bl_{s}(x_{i}).
\end{multlined}
\end{equation}

\begin{claim}
\label{clm:one-point-densities_derivative}
For each $n\ge 1$, 
\begin{equation}
\begin{multlined}
\dv{s}\int_{0}^{\infty}\dd{r}_{1}\dots\int_{0}^{\infty}\dd{r}_{n}\int_{\partial B_{r_{1}}(v)}\dd{x}_{1}\dots\int_{\partial B_{r_{n}}(v)}\dd{x}_{n} e^{-U(\bx)}\prod_{i=1}^{n}\bl_{s}(x_{i})
\\
=n\int_{\partial B_{s}(v)}\dd{w}\bl(w)(e^{-\phi(v-w)}-1)\int_{(\R^{d})^{n-1}}\dd{\bx}e^{-U(\bx,w)}\prod_{i=1}^{n-1}\bl_{s}(x_{i}).
\end{multlined}
\end{equation}
\end{claim}

\begin{proof}
Consider the auxiliary function 
\begin{equation}
f_{n}(s_{1},\dots,s_{n})
:=\int_{0}^{\infty}\dd{r}_{1}\dots\int_{0}^{\infty}\dd{r}_{n}\int_{\partial B_{r_{1}}(v)}\dd{x}_{1}\dots\int_{\partial B_{r_{n}}(v)}\dd{x}_{n} e^{-U(\bx)}\prod_{i=1}^{n}\bl_{s_{i}}(x_{i}).
\end{equation}
By the multivariable chain rule and symmetry, we have that
\begin{equation}
\dv{s}f_{n}(s,\dots,s)
=\sum_{i=1}^{n}\eval{\pdv{s_{i}}f_{n}(s_{1},\dots,s_{n})}_{(s,\dots,s)}
=n\eval{\pdv{s_{n}}f_{n}(s_{1},\dots,s_{n})}_{(s,\dots,s)}
\end{equation}
Since $\bl_{s_{n}}(x_{n})=\left[1+(e^{-\phi(v-x_{n})}-1)\indicator{r_{n}<s_{n}}\right]\bl(x_{n})$ for $x_{n}\in\partial B_{r_{n}}(v)$,
\begin{equation}
\begin{multlined}
\pdv{s_{n}}f_{n}(s_{1},\dots,s_{n})
=\pdv{s_{n}}
\int_{0}^{s_{n}}\dd{r}_{n}\int_{\partial B_{r_{n}}(v)}\dd{x}_{n}(e^{-\phi(v-x_{n})}-1)\bl(x_{n})
\\
\times
\int_{0}^{\infty}\dd{r}_{1}\dots\int_{0}^{\infty}\dd{r}_{n-1}
\int_{\partial B_{r_{1}}(v)}\dd{x}_{1}\dots\int_{\partial B_{r_{n-1}}(v)}\dd{x}_{n-1}
e^{-U(\bx)}\prod_{i=1}^{n-1}\bl_{s_{i}}(x_{i}),
\end{multlined}
\end{equation}
which completes the proof.
\end{proof}

Using Claim \ref{clm:one-point-densities_derivative}, we differentiate $Z(\bl_{s})$ term by term to get that
\begin{equation}
\begin{split}
\dv{s}Z(\bl_{s}){}
&=\sum_{n=1}^{\infty}\frac{1}{n!}\cdot n\int_{\partial B_{s}(v)}\dd{w}\bl(w)(e^{-\phi(v-w)}-1)\int_{(\R^{d})^{n-1}}\dd{\vb{x}} e^{-U(\bx,w)}\prod_{i=1}^{n-1}\bl_{s}(x_{i}) \\
{}&=\int_{\partial B_{s}(v)}\dd{w}\bl(w)(e^{-\phi(v-w)}-1)\sum_{n=0}^{\infty}\frac{1}{n!}\int_{(\R^{d})^{n}}\dd{\vb{x}} e^{-U(\bx)}\prod_{i=1}^{n}\left[\bl_{s}(x_{i})e^{-\phi(w-x_{i})}\right] \\
{}&=\int_{\partial B_{s}(v)}\dd{w}\bl(w)(e^{-\phi(v-w)}-1)Z(\bl_{s}e^{-\phi(w-\cdot)}).
\end{split}
\end{equation}
Thus,
\begin{equation}
\label{eqn:complex-contraction_recursion_preliminary-log-Z-difference}
\begin{multlined}
\log Z(\bl e^{-\phi(v-\cdot)})-\log Z(\bl)
=\int_{0}^{\infty}\dd{s}\frac{\dv{s}Z(\bl_{s})}{Z(\bl_{s})} 
\\
=\int_{0}^{\infty}\dd{s}\int_{\partial B_{s}(v)}\dd{w}\bl(w)(e^{-\phi(v-w)}-1)\frac{Z(\bl_{s}e^{-\phi(w-\cdot)})}{Z(\bl_{s})}.
\end{multlined}
\end{equation}

Suppose first that $s<\hcconst$, and let $w\in\partial B_{s}(v)$.
In this case, $\indicator{B_{d(v,w)}(v)^{c}}(w)=1$, and 
\begin{equation}
\bl_{s}(\cdot)=\bl_{d(v,w)}(\cdot)=\indicator{B_{d(v,w)}(v)^{c}}(\cdot)\bl(\cdot),
\end{equation}
so, after recalling \eqref{eqn:one-point-densities_modulation},
\begin{equation}
\label{eqn:one-point-densities_density-interpolation-integrand}
\bl(w)\frac{Z(\bl_{s}e^{-\phi(w-\cdot)})}{Z(\bl_{s})}
=\bl_{v\rightarrow w}(w)\frac{Z(\bl_{v\rightarrow w}e^{-\phi(w-\cdot)})}{Z(\bl_{v\rightarrow w})}
=\rho_{\bl_{v\rightarrow w}}(w).
\end{equation}

Suppose next that $s\ge\hcconst$, and let $w\in\partial B_{s}(v)$.
In this case, we have that $\indicator{B_{\hcconst}(v)^{c}}(w)=1$, $1+(e^{-\phi(v-w)}-1)\indicator{B_{d(v,w)}(v)}(w)=1$, and 
\begin{equation}
\bl_{s}(\cdot)=\bl_{d(v,w)}(\cdot)=\indicator{B_{\hcconst}(v)^{c}}(\cdot)\left[1+(e^{-\phi(v-\cdot)}-1)\indicator{B_{d(v,w)}(v)}(\cdot)\right]\bl(\cdot),
\end{equation}
so the identity \eqref{eqn:one-point-densities_density-interpolation-integrand} holds true again, after recalling \eqref{eqn:one-point-densities_modulation}.

Therefore,
\begin{equation}
\label{eqn:complex-contraction_recursion_log-Z-difference}
\log Z(\bl e^{-\phi(v-\cdot)})-\log Z(\bl)
=\int_{\R^{d}}\dd{w}\rho_{\bl_{v\rightarrow w}}(w)(e^{-\phi(v-w)}-1).
\end{equation}
The theorem follows after exponentiating \eqref{eqn:complex-contraction_recursion_log-Z-difference} and recalling \eqref{eqn:one-point-densities_def}.
\end{proof}

\section{Tree recursions}
\label{sec:tree-recursions}

In this section, we develop the framework for studying iterations of the recursive identity in Theorem \ref{thm:recursion}.
The recursive identity \eqref{eqn:one-point-densities_recursion} tells us that to compute the one-point density $\rho_{\bl}(v_{0})$ of any given activity function $\bl$, it suffices to compute the one-point density $\rho_{\bl_{v_{0}\rightarrow v_{1}}}(v_{1})$ for each $v_{1}\in\R^{d}$, to which end it suffices to compute the one-point density $\rho_{(\bl_{v_{0}\rightarrow v_{1}})_{v_{1}\rightarrow v_{2}}}(v_{2})$ for each $v_{2}\in\R^{d}$, and so forth.
Our goal is to elucidate the relationship between the initial one-point density $\rho_{\bl}(v_{0})$ and the ones that arise after many iterations of the recursive identity.
We will accomplish this using the formalism of tree recursions introduced by Michelen and Perkins in \cite[\S4]{michelen2023analyticity}.

Throughout the section, we fix an activity function $\bl:\Lambda\rightarrow\C$.
In anticipation of the recursive argument in \S\ref{sec:zero-freeness}, we will work with a general family of activity functions which we call \emph{modulations} of $\bl$.

\begin{definition}[modulation]
\label{def:tree-recursions_modulation}
Let $\alpha\in[0,1]$, $n\ge 0$ an integer, $\set{x_{1},\dots,x_{n}}\subset\R^{d}$ such that $U(x_{1},\dots,x_{n})<\infty$, $t_{1},\dots,t_{n}\in[\hcconst,\infty]$, and $A\subseteq\cap_{i=1}^{n}B_{\hcconst}(x_{i})^{c}$ measurable.
Corresponding to these data is a \emph{modulation} of $\bl$, defined as the activity function
\begin{equation}
\label{eqn:tree-recursions_modulation}
\alpha\indicator{A}(\cdot)\prod_{i=1}^{n}\left[1+(e^{-\phi(x_{i}-\cdot)}-1)\indicator{B_{t_{i}}(x_{i})}(\cdot)\right]\bl(\cdot).
\end{equation}
In the sequel, we will denote by $\bl^{m}$ a generic modulation of $\bl$ and always assume that it is in the form of \eqref{eqn:tree-recursions_modulation}.
\end{definition}

Below, we will study the general problem of characterizing the relationship between a modulation $\bl^{m}$ of $\bl$ and the activity functions that are generated after passing the one-point density $\rho_{\bl^{m}}(v_{0})$, for any $v_{0}\in\R^{d}$, through multiple iterations of the recursive identity \eqref{eqn:one-point-densities_recursion}.
Notice that the activity function $\bl$ is a modulation of itself with $\alpha=1$, $n=0$, and $A=\R^{d}$, so this general program indeed encompasses our original goal.

First, we recall that the validity of the recursive identity depends on certain interpolating activity functions to have nonzero partition functions (see Theorem \ref{thm:recursion}).
As promised, we now introduce a notion of zero-freeness for an activity function that covers all such interpolations.

\begin{definition}[total zero-freeness]
\label{def:tree-recursions_totally-zero-free}
We say that an activity function $\bl:\Lambda\rightarrow\C$ is \emph{totally zero-free} if $Z(\bl^{m})\ne 0$ for all modulations $\bl^{m}$ of $\bl$.
\end{definition}

Next, we import some basic notation and terminology from \cite[\S4.1]{michelen2023analyticity}.

\begin{definition}[tree recursion]
Let $k\ge 1$.
By a \emph{depth-$k$ modulating function} (originally called a \emph{depth-$k$ damping function} in \cite{michelen2023analyticity}), we mean a bounded, measurable function $\vec{\gamma}:\cup_{j=1}^{k}(\R^{d})^{j+1}\rightarrow\R_{\ge 0}$.
By a \emph{depth-$k$ boundary condition}, we mean a bounded, measurable function $\vec{\tau}:(\R^{d})^{k+1}\rightarrow\C$.
Given a modulation $\bl^{m}$ of $\bl$, a depth-$k$ boundary condition $\vec{\tau}$, and a depth-$k$ modulating function $\vec{\gamma}$, we define the corresponding \emph{depth-$k$ tree recursion} as the function $\pi_{\bl,\vec{\tau},\vec{\gamma}}:\cup_{j=1}^{k}(\R^{d})^{j+1}\rightarrow\C$ given by
\begin{equation}
\label{eqn:tree-recursions_basics_tree-recursion}
\begin{multlined}
\pi_{\bl^{m},\vec{\tau},\vec{\gamma}}(v_{0},\dots,v_{j-1}):=
\bl^{m}(v_{j-1})\vec{\gamma}(v_{0},\dots,v_{j-1},v_{j-1})
\\
\times\exp{-\int_{\R^{d}}\dd{w}\left[1-e^{-\phi(v_{j-1}-w)}\right]\pi_{\bl^{m},\vec{\tau},\vec{\gamma}}(v_{0},\dots,v_{j-1},w)}
\end{multlined}
\end{equation}
for $1\le j\le k$, and
\begin{equation}
\pi_{\bl^{m},\vec{\tau},\vec{\gamma}}(v_{0},\dots,v_{k}):=\vec{\tau}(v_{0},\dots,v_{k}).
\end{equation}
\end{definition}

\begin{remark}
In this context, the term \emph{boundary condition} refers to a value that the recursion returns after the maximum depth is reached, rather than some particle configuration in the physical sense of the term.
\end{remark}

\begin{remark}
In the original work of Michelen and Perkins \cite{michelen2023analyticity}, the damping function $\vec{\gamma}$ provides a multiplicative factor of $\vec{\gamma}(v_{0},\dots,v_{j-1})$ in an intermediate step \eqref{eqn:tree-recursions_basics_tree-recursion} in the tree recursion rather than our factor of $\vec{\gamma}(v_{0},\dots,v_{j-1},v_{j-1})$.
Ours is the result of an intentional change of notation that, we think, better clarifies each entry's role in the modulating function: the last entry is the one true argument of the function whereas all the preceding ones only serve to determine its form.
\end{remark}

Let $k\ge 1$ and $v_{0}\in\R^{d}$.
In the remainder of the section, we show that we can always design a depth-$k$ tree recursion $\pi_{\bl^{m},\vec{\tau},\vec{\gamma}}(v_{0})$ that computes the one-point density at $v_{0}$ at activity $\bl^{m}$, i.e., $\rho_{\bl^{m}}(v_{0})=\pi_{\bl^{m},\vec{\tau},\vec{\gamma}}(v_{0})$, by choosing appropriately the (depth-$k$) boundary condition $\vec{\tau}$ and modulating function $\vec{\gamma}$.

In light of the form of the activity function \eqref{eqn:one-point-densities_modulation} generated by one iteration of the recursive identity, we start by classifying sequences of points in $\R^{d}$ of each finite length $j\ge 1$ as follows.
For each $I\subseteq\set{0,\dots,j-2}$, let
\begin{equation}
\label{eqn:tree-recursions_good-sequences}
\begin{multlined}
G_{I}^{(j)}:=\left\{(v_{0},\dots,v_{j-1})\in(\R^{d})^{j}\mid d(v_{i},v_{i+1})\ge\hcconst\text{ if and only if }i\in I,\right.
\\
\left.\text{and }d(v_{i},v_{i'})\ge\hcconst\text{ for all }i\in I\text{ and }i<i'\le j-1\vphantom{(\R^{d})^{j}}\right\},
\end{multlined}
\end{equation}
and set
\begin{align}
G^{(j)}{}&:=\bigcup_{I\subseteq\set{0,\dots,j-2}}G^{(j)}_{I},\\
\label{eqn:tree-recursions_bad-sequences}
B^{(j)}{}&:=(\R^{d})^{j}\setminus G^{(j)}.
\end{align}
Notice that the second part of the defining condition of $G^{(j)}_{I}$ in \eqref{eqn:tree-recursions_good-sequences} implies that $U(v_{i}\mid i\in I)<\infty$.

\begin{definition}[reference modulating function]
\label{def:tree-recursions_canonical-modulating-function}
Using the above classification scheme, we define the \emph{reference modulating function} $\vec{\gamma}_{c}:\cup_{j=1}^{\infty}(\R^{d})^{j+1}\rightarrow\R_{\ge 0}$ by
\begin{equation}
\label{eqn:tree-recursions_canonical-modulating-function-good-sequence}
\begin{multlined}
\vec{\gamma}_{c}(v_{0},\dots,v_{j-1},w):=
\prod_{\substack{i=0\\ i\in I}}^{j-2}\left\{\indicator{B_{\hcconst}(v_{i})^{c}}(w)\left[(e^{-\phi(w-v_{i})}-1)\indicator{B_{d(v_{i},v_{i+1})}(v_{i})}(w)+1\right]\right\}
\\
\times\prod_{\substack{i=0\\ i\not\in I}}^{j-2}\indicator{B_{d(v_{i},v_{i+1})}(v_{i})^{c}}(w)
\text{,\quad if }(v_{0},\dots,v_{j-1})\in G_{I}^{(j)},
\end{multlined}
\end{equation}
and 
\begin{equation}
\label{eqn:tree-recursions_canonical-modulating-function-bad-sequence}
\vec{\gamma}_{c}(v_{0},\dots,v_{j-1},w):=0\text{,\quad if } (v_{0},\dots,v_{j-1})\in B^{(j)},
\end{equation}
for $j\ge 1$.
\end{definition}

We are now ready to specify the functions $\vec{\tau}$ and $\vec{\gamma}$ suitable for describing iterations of the recursive identity \eqref{eqn:one-point-densities_recursion}.

\begin{theorem}
\label{thm:tree-recursions_density-correspondence}
Let $k\ge 1$, $v_{0}\in\R^{d}$, and $\bl:\Lambda\rightarrow\C$ a totally zero-free activity function.
Let $\bl^{m}$ be a modulation of $\bl$.
Define the depth-$k$ modulating function 
\begin{equation}
\label{eqn:tree-recursions_modulating-function-choice}
\vec{\gamma}(v_{0},\dots,v_{j-1},w):=\indicator{\set{v_{0},\dots,v_{j-1}}\subseteq A}\vec{\gamma}_{c}(v_{0},\dots,v_{j-1},w),
\end{equation}
where the set $A\subseteq\R^{d}$ is supplied by $\bl^{m}$ as in \eqref{eqn:tree-recursions_modulation}, and the depth-$k$ boundary condition 
\begin{equation}
\label{eqn:tree-recursions_boundary-condition-choice}
\vec{\tau}(v_{0},\dots,v_{k}):=
\begin{cases}
\rho_{\bl^{m}_{v_{0}\rightarrow\dots\rightarrow v_{k}}}(v_{k}) & \text{if } (v_{0},\dots,v_{k})\in G^{(k+1)} \\
0 & \text{else}
\end{cases},
\end{equation}
where we introduce the shorthand
\begin{equation}
\label{eqn:tree-recursions_shorthand}
\bl^{m}_{v_{0}\rightarrow\dots\rightarrow v_{k}}(\cdot)
:=\vec{\gamma}(v_{0},\dots,v_{k},\cdot)\bl(\cdot).
\end{equation}
Then, the corresponding depth-$k$ tree recursion $\pi_{\bl^{m},\vec{\tau},\vec{\gamma}}$ satisfies
\begin{equation}
\label{eqn:tree-recursions_stronger-recursion-density-relation}
\pi_{\bl^{m},\vec{\tau},\vec{\gamma}}(v_{0},\dots,v_{j-1})=
\begin{cases}
\rho_{\bl^{m}_{v_{0}\rightarrow\dots\rightarrow v_{j-1}}}(v_{j-1}) & \text{if }(v_{0},\dots,v_{j-1})\in G^{(j)} \\
0 & \text{if }(v_{0},\dots,v_{j-1})\in B^{(j)}
\end{cases}
\end{equation}
for all $1\le j\le k+1$.
In particular, the tree recursion $\pi_{\bl^{m},\vec{\tau},\vec{\gamma}}$ computes the one-point density at $v_{0}$ at activity $\bl^{m}$ in the sense that $\pi_{\bl^{m},\vec{\tau},\vec{\gamma}}(v_{0})=\rho_{\bl^{m}}(v_{0})$.
\end{theorem}

\begin{proof}
To prove \eqref{eqn:tree-recursions_stronger-recursion-density-relation}, we use induction on $k+1-j$.
When $j=k+1$, the tree recursion returns the boundary condition by definition, i.e., $\pi_{\bl^{m},\vec{\tau},\vec{\gamma}}(v_{0},\dots,v_{k})=\vec{\tau}(v_{0},\dots,v_{k})$.
By \eqref{eqn:tree-recursions_boundary-condition-choice}, there is nothing to prove.

Suppose now that $j\le k$.
We carry out one step of the tree recursion as in \eqref{eqn:tree-recursions_basics_tree-recursion}:
\begin{equation}
\label{eqn:tree-recursions_density-correspondence_one-step-tree-recursion}
\begin{multlined}
\pi_{\bl^{m},\vec{\tau},\vec{\gamma}}(v_{0},\dots,v_{j-1})=
\bl^{m}(v_{j-1})\vec{\gamma}(v_{0},\dots,v_{j-1},v_{j-1})
\\
\times\exp{-\int_{\R^{d}}\dd{w}\left[1-e^{-\phi(v_{j-1}-w)}\right]\pi_{\bl^{m},\vec{\tau},\vec{\gamma}}(v_{0},\dots,v_{j-1},w)}.
\end{multlined}
\end{equation}
If $(v_{0},\dots,v_{j-1})\in B^{(j)}$, then by the construction of the modulating function $\vec{\gamma}$ (see \eqref{eqn:tree-recursions_modulating-function-choice} and \eqref{eqn:tree-recursions_canonical-modulating-function-bad-sequence}), 
$\vec{\gamma}(v_{0},\dots,v_{j-1},v_{j-1})=0$, so $\pi_{\bl^{m},\vec{\tau},\vec{\gamma}}(v_{0},\dots,v_{j-1})=0$ as required.
Thus, we may assume that $(v_{0},\dots,v_{j-1})\in G^{(j)}_{I}$, where $I\subseteq\set{0,\dots,j-2}$.
The idea now is to compare the tree recursion identity \eqref{eqn:tree-recursions_density-correspondence_one-step-tree-recursion} to the result of applying the recursive identity \eqref{eqn:one-point-densities_recursion} to $\rho_{\bl^{m}_{v_{0}\rightarrow\dots\rightarrow v_{j-1}}}(v_{j-1})$.
We will accomplish this in several steps.

Firstly, we apply the inductive hypothesis to $\pi_{\bl^{m},\vec{\tau},\vec{\gamma}}(v_{0},\dots,v_{j-1},w)$.

\begin{claim}
\label{clm:tree-recursions_density-correspondence_bad-sequence-condition}
Suppose that $(v_{0},\dots,v_{j-1})\in G^{(j)}$.
For any $w\in\R^{d}$, $(v_{0},\dots,v_{j-1},w)\in B^{(j+1)}$ if and only if $w\in\cup_{i\in I}B_{\hcconst}(v_{i})$.
\end{claim}

\begin{proof}
This is clear from the definitions \eqref{eqn:tree-recursions_good-sequences} and \eqref{eqn:tree-recursions_bad-sequences}.
\end{proof}

By Claim \ref{clm:tree-recursions_density-correspondence_bad-sequence-condition} and the inductive hypothesis, we get from \eqref{eqn:tree-recursions_density-correspondence_one-step-tree-recursion} that 
\begin{equation}
\label{eqn:tree-recursions_density-correspondence_first-rewriting}
\begin{multlined}
\pi_{\bl^{m},\vec{\tau},\vec{\gamma}}(v_{0},\dots,v_{j-1})=
\bl^{m}(v_{j-1})\vec{\gamma}(v_{0},\dots,v_{j-1},v_{j-1})
\\
\times\exp{-\int_{\R^{d}\setminus\cup_{i\in I}B_{\hcconst}(v_{i})}\dd{w}\left[1-e^{-\phi(v_{j-1}-w)}\right]\rho_{\bl^{m}_{v_{0}\rightarrow\dots\rightarrow v_{j-1}\rightarrow w}(w)}}.
\end{multlined}
\end{equation}

Secondly, we use the recursive identity \eqref{eqn:one-point-densities_recursion} to compute the one-point density $\rho_{\bl^{m}_{v_{0}\rightarrow\dots\rightarrow v_{j-1}}}(v_{j-1})$.
To do this, we must check that the assumption of Theorem \ref{thm:recursion} is met.

\begin{claim}
\label{clm:tree-recursions_density-correspondence_interpolation}
For all $s\ge 0$, the activity function 
\begin{equation}
\label{eqn:tree-recursions_density-correspondence_interpolation}
\bl_{s}(\cdot)=\left[1+(e^{-\phi(v_{j-1}-\cdot)}-1)\indicator{B_{s}(v_{j-1})}(\cdot)\right]\bl^{m}_{v_{0}\rightarrow\dots\rightarrow v_{j-1}}(\cdot)
\end{equation}
is a modulation of $\bl$.
Consequently, $Z(\bl_{s})\ne 0$ by the total zero-freeness of $\bl$.
\end{claim}

\begin{proof}
Let $s\ge 0$.
By the radial symmetry of the potential $\phi$, we rewrite 
\begin{equation}
\begin{multlined}
1+(e^{-\phi(v_{j-1}-\cdot)}-1)\indicator{B_{s}(v_{j-1})}(\cdot)
\\
=\begin{cases}
\indicator{B_{s}(v_{j-1})^{c}}(\cdot) & s<\hcconst \\
\indicator{B_{\hcconst}(v_{j-1})^{c}}(\cdot)\left[1+(e^{-\phi(v_{j-1}-\cdot)}-1)\indicator{B_{s}(v_{j-1})}(\cdot)\right] & s\ge\hcconst
\end{cases}.
\end{multlined}
\end{equation}
On the other hand, using \eqref{eqn:tree-recursions_shorthand} and \eqref{eqn:tree-recursions_modulation}, we have that
\begin{equation}
\label{eqn:tree-recursions_density-correspondence_interpolation_full-expansion}
\begin{multlined}
\bl^{m}_{v_{0}\rightarrow\dots\rightarrow v_{j-1}}(\cdot)
=\alpha
\indicator{\set{v_{0},\dots,v_{j-1}}\subseteq A}
\indicator{A'}(\cdot)
\prod_{\substack{i=0\\ i\in I}}^{j-2} 
\left[(e^{-\phi(v_{i}-\cdot)}-1)\indicator{B_{d(v_{i},v_{i+1})}(v_{i})}(\cdot)+1\right] 
\\
\times
\prod_{i=1}^{n}\left[1+(e^{-\phi(x_{i}-\cdot)}-1)\indicator{B_{t_{i}}(x_{i})}(\cdot)\right]\bl(\cdot),
\end{multlined}
\end{equation}
where 
\begin{equation}
A':=A\cap\bigcap_{\substack{i=0\\ i\in I}}^{j-2}B_{\hcconst}(v_{i})^{c}\cap\bigcap_{\substack{i=0\\ i\not\in I}}^{j-2}B_{d(v_{i},v_{i+1})}(v_{i})^{c}.
\end{equation}
If $\set{v_{0},\dots,v_{j-1}}\not\subseteq A$, then $\bl^{m}_{v_{0}\rightarrow\dots\rightarrow v_{j-1}}$ is the zero activity function, which is trivially a modulation of $\bl$ (with $\alpha=0$).
Else, we have that 
\begin{equation}
\label{eqn:tree-recursions_density-correspondence_interpolation_finite-potential-energy}
U(x_{1},\dots,x_{n},v_{i}\mid i\in I)<\infty
\end{equation}
since $A\subseteq\cap_{i=1}^{n}B_{\hcconst}(x_{i})^{c}$ and $(v_{0},\dots,v_{j-1})\in G^{(j)}_{I}$.
By Definition \ref{def:tree-recursions_modulation}, $\bl^{m}_{v_{0}\rightarrow\dots\rightarrow v_{j-1}}(\cdot)$ is again a modulation of $\bl$.
\end{proof}

By Theorem \ref{thm:recursion}, we have that
\begin{equation}
\label{eqn:tree-recursions_density-correspondence_density-recurion}
\begin{multlined}
\rho_{\bl^{m}_{v_{0}\rightarrow\dots\rightarrow v_{j-1}}}(v_{j-1})
=\bl^{m}_{v_{0}\rightarrow\dots\rightarrow v_{j-1}}(v_{j-1})
\\
\times\exp{-\int_{\R^{d}}\dd{w}\left[1-e^{-\phi(v_{j-1}-w)}\right]\rho_{(\bl^{m}_{v_{0}\rightarrow\dots\rightarrow v_{j-1}})_{v_{j-1}\rightarrow w}}(w)}.
\end{multlined}
\end{equation}
Seeing that $\bl^{m}_{v_{0}\rightarrow\dots\rightarrow v_{j-1}}(v_{j-1})=\bl^{m}(v_{j-1})\vec{\gamma}(v_{0},\dots,v_{j-1},v_{j-1})$, it remains to rewrite the density in the integral in the same form as in \eqref{eqn:tree-recursions_density-correspondence_first-rewriting}.

\begin{claim}
Suppose that $(v_{0},\dots,v_{j-1})\in G^{(j)}_{I}$, where $I\subseteq\set{0,\dots,j-2}$.
For any $w\in\R^{d}$,
\begin{equation}
\label{eqn:tree-recursions_density-correspondence_density-rewriting}
\rho_{(\bl^{m}_{v_{0}\rightarrow\dots\rightarrow v_{j-1}})_{v_{j-1}\rightarrow w}}(w)
=\begin{cases}
\rho_{\bl^{m}_{v_{0}\rightarrow\dots\rightarrow v_{j-1}\rightarrow w}}(w) & \text{if }w\in\R^{d}\setminus\bigcup_{i\in I}B_{\hcconst}(v_{i}) \\
0 & \text{else}
\end{cases}.
\end{equation}
\end{claim}

\begin{proof}
By \eqref{eqn:tree-recursions_shorthand} and \eqref{eqn:one-point-densities_modulation}, we write
\begin{equation}
(\bl^{m}_{v_{0}\rightarrow\dots\rightarrow v_{j-1}})_{v_{j-1}\rightarrow w}(\cdot)
=\bl^{m}(\cdot)\vec{\gamma}(v_{0}\dots,v_{j-1},\cdot)\vec{\gamma}(v_{j-1},w,\cdot).
\end{equation}
If $w\in\cup_{i\in I}B_{\hcconst}(v_{i})$, then $\vec{\gamma}(v_{0}\dots,v_{j-1},w)=0$, so $\rho_{(\bl^{m}_{v_{0}\rightarrow\dots\rightarrow v_{j-1}})_{v_{j-1}\rightarrow w}}(w)=0$ by \eqref{eqn:one-point-densities_def}.
Otherwise, we have that
\begin{equation}
(v_{0},\dots,v_{j-1},w)\in
\begin{cases}
G^{(j+1)}_{I} & \text{if }d(v_{j-1},w)<\hcconst \\
G^{(j+1)}_{I\cup\set{j-1}} & \text{else}
\end{cases},
\end{equation}
so 
\begin{equation}
\bl^{m}(\cdot)\vec{\gamma}(v_{0}\dots,v_{j-1},\cdot)\vec{\gamma}(v_{j-1},w,\cdot)
=\bl^{m}(\cdot)\vec{\gamma}(v_{0},\dots,v_{j-1},w,\cdot)
=\bl^{m}_{v_{0}\rightarrow\dots\rightarrow v_{j-1}\rightarrow w}(\cdot),
\end{equation}
which completes the proof.
\end{proof}

The theorem follows after inserting \eqref{eqn:tree-recursions_density-correspondence_density-rewriting} into \eqref{eqn:tree-recursions_density-correspondence_density-recurion} and comparing the latter with \eqref{eqn:tree-recursions_density-correspondence_first-rewriting}.
\end{proof}

Finally, we note a boundedness property of the modulating function $\vec{\gamma}$ chosen in Theorem \ref{thm:tree-recursions_density-correspondence}, which will be essential in the coming analysis.

\begin{lemma}
\label{lem:one-point-densities_uniform-boundedness}
Let $\bl:\Lambda\rightarrow\C$ be an activity function (which is bounded by definition) and $\bl^{m}$ a modulation of $\bl$.
For all $k\ge 1$, $1\le j\le k$, and $v_{0},\dots,v_{j-1}\in\R^{d}$, the depth-$k$ modulating function $\vec{\gamma}$ in Theorem \ref{thm:tree-recursions_density-correspondence} satisfies 
\begin{equation}
\label{eqn:one-point-densities_uniform-boundedness}
\abs{\vec{\gamma}(v_{0},\dots,v_{j-1},\cdot)\bl^{m}(\cdot)}\le e^{\lsconst}\norm{\bl}_{\infty},
\end{equation}
where $\lsconst\ge 0$ is the local stability constant defined in \eqref{eqn:introduction_local-stability}.
\end{lemma}

\begin{proof}
If $(v_{0},\dots,v_{j-1})\in B^{(j)}$, then $\vec{\gamma}(v_{0},\dots,v_{j-1},\cdot)$ is identically zero and there is nothing to prove.
Otherwise, we fully expand $\vec{\gamma}(v_{0},\dots,v_{j-1},\cdot)\bl^{m}(\cdot)$ as in \eqref{eqn:tree-recursions_density-correspondence_interpolation_full-expansion}.
We may assume that $\set{v_{0},\dots,v_{j-1}}\subseteq A\subseteq\cap_{i=1}^{n}B_{\hcconst}(x_{i})^{c}$ (see \eqref{eqn:tree-recursions_density-correspondence_interpolation_full-expansion}).
By the same argument as in the proof of Claim \ref{clm:tree-recursions_density-correspondence_interpolation}, we deduce the finite energy property \eqref{eqn:tree-recursions_density-correspondence_interpolation_finite-potential-energy}.
Therefore, no matter which indicator functions in \eqref{eqn:tree-recursions_density-correspondence_interpolation_full-expansion} are nonzero, the bound \eqref{eqn:one-point-densities_uniform-boundedness} holds by local stability.
\end{proof}

\section{Contraction properties}
\label{sec:contraction-properties}

In this section, we analyze the contraction properties of finite-depth tree recursions.
Our central object of study is the functional 
\begin{equation}
\label{eqn:contraction-properties_functional}
F_{v}(\bl,\br):=\bl(v)\exp{-\int_{\R^{d}}\dd{w}\left[1-e^{-\phi(v-w)}\right]\br(w)}
\end{equation}
abstracted from the recursive identity \eqref{eqn:one-point-densities_recursion} which corresponds to a single step in a tree recursion.
We will start with the simple case of \emph{real-valued} activity and one-point density functions (i.e., $\bl$ and $\br$) to illustrate the main ideas of the analysis, before moving on to the general and much more technical case of all \emph{complex-valued} functions as is required for proving the analyticity of the pressure.

\subsection{Real contraction}
\label{sec:contraction-properties_real-contraction}

First, we consider the functional \eqref{eqn:contraction-properties_functional} as defined over bounded, measurable functions $\bl,\br:\R^{d}\rightarrow\R_{\ge 0}$.
Following the treatment in \cite[\S5.1]{michelen2023analyticity}, we apply a change of coordinates \cite{guo2018uniqueness,li2013correlation} to the functional $F_{v}$ through the \emph{potential function} $\psi(x):=\sqrt{x}$, namely, we will study the transformed functional
\begin{equation}
g_{v}(\bl,\bz)
:=\psi(F_{v}(\bl,\psi^{-1}(\bz)))
=\left[\bl(v)\exp{-\int_{\R^{d}}\dd{w}\left[1-e^{-\phi(v-w)}\right]\bz(w)^{2}}\right]^{1/2}.
\end{equation}
From here on, we restrict ourselves to only those $\bl$ and $\bz$ that take values in $[0,\tl]$ and $[0,\tz]$, respectively, where $\tilde{\lambda},\tilde{z}\ge0$ are constants that we will choose optimally later.

We will make extensive use of the following estimate of the difference in the outputs of the functional when applied to the same activity function $\bl$ but different $\bz$'s.

\begin{lemma}
\label{lem:contraction-properties_real-contraction_single-iteration}
Let $\tilde{z}>0$.
Denote by $M(\tilde{z})$ the maximum of the function
\begin{equation}
\label{eqn:contraction-properties_real-contraction_single-iteration_coefficient-function}
z\mapsto e^{-\pconst z^{2}+\aconst \tilde{z}^{2}}(\pconst z^{2}+\aconst \tilde{z}^{2})
\end{equation} 
on $[0,\tilde{z}]$.
For any Lebesgue measurable functions $\bx,\by:\R^{d}\rightarrow[0,\tilde{z}]$,
\begin{equation}
\abs{g_{v}(\bl,\bx)-g_{v}(\bl,\by)}^{2}
\le\bl(v)M(\tilde{z})\int_{\R^{d}}\dd{w}\abs{1-e^{-\phi(v-w)}}\abs{\bx(w)-\by(w)}^{2}.
\end{equation}
\end{lemma}

\begin{proof}
Let $\bz_{s}:=s\bx+(1-s)\by$.
By the mean value theorem, there exists $t\in(0,1)$ such that
\begin{equation}
\begin{multlined}
\abs{g_{v}(\bl,\bx)-g_{v}(\bl,\by)}
\le\abs{\eval{\dv{s}g_{\lambda}(\bl,\bz_{s})}_{s=t}}
\le\bl(v)^{1/2}
\\
\times\exp{-\frac{1}{2}\int_{\R^{d}}\dd{w}\left[1-e^{-\phi(v-w)}\right]\bz_{t}(w)^{2}}
\int_{\R^{d}}\dd{w}\abs{1-e^{-\phi(v-w)}}\bz_{t}(w)\abs{\by(w)-\bx(w)}.
\end{multlined}
\end{equation}
Using the Cauchy-Schwarz inequality with the measure $\dd{w}\abs{1-e^{-\phi(v-w)}}$ on $\R^{d}$ on the last integral, we get that
\begin{equation}
\label{eqn:contraction-properties_real-contraction_single-iteration_intermediate-bound}
\begin{multlined}
\abs{g_{v}(\bl,\bx)-g_{v}(\bl,\by)}^{2}
\le \bl(v)\exp{-\int_{\R^{d}}\dd{w}\left[1-e^{-\phi(v-w)}\right]\bz_{t}(w)^{2}}
\\
\times \int_{\R^{d}}\dd{w}\abs{1-e^{-\phi(v-w)}}\bz_{t}(w)^{2}
\cdot\int_{\R^{d}}\dd{w}\abs{1-e^{-\phi(v-w)}}\abs{\by(w)-\bx(w)}^{2}.
\end{multlined}
\end{equation}
Now, let $\mathcal{D}_{v}:=\set{w\in\R^{d}\mid\phi(v-w)\ge 0}$.
By the convexity of the image of $[0,\tilde{z}]$ under the mapping $z\mapsto z^{2}$, there exist $z_{1},z_{2}\in[0,\tilde{z}]$ such that
\begin{equation}
\label{eqn:contraction-properties_real-contraction_single-iteration_convexity}
\pconst z_{1}^{2}=\int_{\mathcal{D}_{v}}\dd{w}\left[1-e^{-\phi(v-w)}\right]\bz_{t}(w)^{2}, \quad
\aconst z_{2}^{2}=\int_{\R^{d}\setminus\mathcal{D}_{v}}\dd{w}\left[e^{-\phi(v-w)}-1\right]\bz_{t}(w)^{2}.
\end{equation}
Consequently, 
\begin{equation}
\label{eqn:contraction-properties_real-contraction_single-iteration_convexity-bound}
\begin{multlined}
\exp{-\int_{\R^{d}}\dd{w}\left[1-e^{-\phi(v-w)}\right]\bz_{t}(w)^{2}}
\int_{\R^{d}}\dd{w}\abs{1-e^{-\phi(v-w)}}\bz_{t}(w)^{2}
\\
=e^{-\pconst z_{1}^{2}+\aconst z_{2}^{2}}(\pconst z_{1}^{2}+\aconst z_{2}^{2})
\le M(\tilde{z}).
\end{multlined}
\end{equation}
The lemma follows after inserting \eqref{eqn:contraction-properties_real-contraction_single-iteration_convexity-bound} into \eqref{eqn:contraction-properties_real-contraction_single-iteration_intermediate-bound}.
\end{proof}

Next, we derive a condition for Lemma \ref{lem:contraction-properties_real-contraction_single-iteration} to be useful toward understanding the contraction properties of tree recursions which, abstractly, involve many successive iterations of the functional.

\begin{lemma}
\label{lem:contraction-properties_real-contraction_self-map-preservation}
Suppose that $(\tilde{\lambda},\tilde{z})\in\R_{\ge 0}^{2}$ satisfies the \emph{self-map condition}
\begin{equation}
\label{eqn:contraction-properties_real-contraction_self-map-condition}
\tilde{\lambda}e^{\aconst\tilde{z}^{2}}\le\tilde{z}^{2}.
\end{equation}
Let $k\ge 1$.
Suppose that $\pi_{\bl,\vec{\tau},\vec{\gamma}}$ is a depth-$k$ tree recursion such that $\vec{\tau}$ takes values in $[0,\tz^{2}]$ and $\vec{\gamma}(v_{0},\dots,v_{j-1},\cdot)\bl(\cdot)$ in $[0,\tl]$ for all $1\le j\le k$ and $v_{0},\dots,v_{j-1}\in\R^{d}$.
Then, for all $1\le j\le k+1$ and $v_{0},\dots,v_{j-1}\in\R^{d}$, $\pi_{\bl,\vec{\tau},\vec{\gamma}}(v_{0},\dots,v_{j-1})\in[0,\tz^{2}]$.
\end{lemma}

\begin{proof}
We use induction on $k+1-j$.
When $j=k+1$, $\pi_{\bl,\vec{\tau},\vec{\gamma}}(v_{0},\dots,v_{k})=\vec{\tau}(v_{0},\dots,v_{k})\in[0,\tz^{2}]$.
Suppose now that $j\le k$.
We carry out one step in the tree recursion as follows:
\begin{equation}
\label{lem:contraction-properties_real-contraction_self-map-preservation_one-step}
\begin{multlined}
\pi_{\bl,\vec{\tau},\vec{\gamma}}(v_{0},\dots,v_{j-1})=
\bl(v_{j-1})\vec{\gamma}(v_{0},\dots,v_{j-1},v_{j-1})
\\
\times\exp{-\int_{\R^{d}}\dd{w}\left[1-e^{-\phi(v_{j-1}-w)}\right]\pi_{\bl,\vec{\tau},\vec{\gamma}}(v_{0},\dots,v_{j-1},w)}.
\end{multlined}
\end{equation}
By the inductive hypothesis, $\pi_{\bl,\vec{\tau},\vec{\gamma}}(v_{0},\dots,v_{j-1},w)\in[0,\tz^{2}]$.
Hence, we can bound 
\begin{equation}
\pi_{\bl,\vec{\tau},\vec{\gamma}}(v_{0},\dots,v_{j-1})\le\tl e^{\aconst\tz^{2}}\le\tz^{2},
\end{equation}
where we used the self-map condition \eqref{eqn:contraction-properties_real-contraction_self-map-condition} in the last inequality.
\end{proof}

Under the self-map condition, we iterate the bound in Lemma \ref{lem:contraction-properties_real-contraction_single-iteration} for finite-depth tree recursions.

\begin{proposition}
\label{prop:contraction-properties_real-contraction_multiple-iteration}
Suppose that $(\tilde{\lambda},\tilde{z})\in\R_{\ge 0}^{2}$ satisfies the self-map condition \eqref{eqn:contraction-properties_real-contraction_self-map-condition}.
Let $k\ge 1$.
Suppose that $\pi_{\bl,\vec{\tau}_{1},\vec{\gamma}}$ and $\pi_{\bl,\vec{\tau}_{2},\vec{\gamma}}$ are two depth-$k$ tree recursions possibly differing in their boundary conditions, where $\bl$ and $\vec{\gamma}$ satisfy the same conditions as in Lemma \ref{lem:contraction-properties_real-contraction_self-map-preservation}, and both boundary conditions $\vec{\tau}_{i}$ take values in $[0,\tz^{2}]$.
Then, for all $1\le j\le k+1$ and $v_{0},\dots,v_{j-1}\in\R^{d}$,
\begin{equation}
\label{eqn:contraction-properties_real-contraction_multiple-iteration}
\begin{multlined}
\abs{\sqrt{\pi_{\bl,\vec{\tau}_{1},\vec{\gamma}}(v_{0},\dots,v_{j-1})}
-\sqrt{\pi_{\bl,\vec{\tau}_{2},\vec{\gamma}}(v_{0},\dots,v_{j-1})}}^{2}
\le M(\tilde{z})^{k+1-j}\int_{(\R^{d})^{k+1-j}}\dd{v}_{j}\dots\dd{v}_{k}
\\
\times \prod_{\ell=j}^{k}\left[\bl(v_{\ell-1})\vec{\gamma}(v_{0},\dots,v_{\ell-1},v_{\ell-1})\vphantom{\abs{1-e^{-\phi(v_{\ell-1}-v_{\ell})}}}\abs{1-e^{-\phi(v_{\ell-1}-v_{\ell})}}\right]
\abs{\sqrt{\vec{\tau}_{1}(v_{0},\dots,v_{k})}-\sqrt{\vec{\tau}_{2}(v_{0},\dots,v_{k})}}^{2}.
\end{multlined}
\end{equation}
\end{proposition}

\begin{proof}
We use induction on $k+1-j$.
When $j=k+1$, the bound \eqref{eqn:contraction-properties_real-contraction_multiple-iteration} asserts that 
\begin{equation}
\abs{\sqrt{\pi_{\bl,\vec{\tau}_{1},\vec{\gamma}}(v_{0},\dots,v_{k})}
-\sqrt{\pi_{\bl,\vec{\tau}_{2},\vec{\gamma}}(v_{0},\dots,v_{k})}}^{2}
\le \abs{\sqrt{\vec{\tau}_{1}(v_{0},\dots,v_{k})}-\sqrt{\vec{\tau}_{2}(v_{0},\dots,v_{k})}}^{2},
\end{equation}
which holds trivially since $\pi_{\bl,\vec{\tau}_{i},\vec{\gamma}}(v_{0},\dots,v_{k})=\vec{\tau}_{i}(v_{0},\dots,v_{k})$, $i=1,2$.
Suppose now that $j\le k$.
Carrying out one step in either tree recursion, we see that
\begin{equation}
\begin{split}
{}&\sqrt{\pi_{\bl,\vec{\tau}_{i},\vec{\gamma}}(v_{0},\dots,v_{j-1})}
\\
={}&
\begin{multlined}[t]
\left[\bl(v_{j-1})\vec{\gamma}(v_{0},\dots,v_{j-1},v_{j-1})
\vphantom{\times\exp{-\int_{\R^{d}}\dd{v}_{j}\left[1-e^{-\phi(v_{j-1}-v_{j})}\right]\pi_{\bl,\vec{\tau}_{i},\vec{\gamma}}(v_{0},\dots,v_{j-1},v_{j})}}
\right.
\\
\left.\times\exp{-\int_{\R^{d}}\dd{v}_{j}\left[1-e^{-\phi(v_{j-1}-v_{j})}\right]\pi_{\bl,\vec{\tau}_{i},\vec{\gamma}}(v_{0},\dots,v_{j-1},v_{j})}\right]^{1/2}
\end{multlined}
\\
={}&g_{v_{j-1}}(\bl(\cdot)\vec{\gamma}(v_{0},\dots,v_{j-1},\cdot),\sqrt{\pi_{\bl,\vec{\tau}_{i},\vec{\gamma}}(v_{0},\dots,v_{j-1},\cdot)}).
\end{split}
\end{equation}
By Lemma \ref{lem:contraction-properties_real-contraction_self-map-preservation}, we can use Lemma \ref{lem:contraction-properties_real-contraction_single-iteration} to deduce that
\begin{equation}
\label{eqn:contraction-properties_real-contraction_multiple-iteration_inductive-step}
\begin{multlined}
\abs{\sqrt{\pi_{\bl,\vec{\tau}_{1},\vec{\gamma}}(v_{0},\dots,v_{j-1})}
-\sqrt{\pi_{\bl,\vec{\tau}_{2},\vec{\gamma}}(v_{0},\dots,v_{j-1})}}^{2}
\le \bl(v_{j-1})
\vec{\gamma}(v_{0},\dots,v_{j-1},v_{j-1})
M(\tz)
\\
\times
\int_{\R^{d}}\dd{v}_{j}\abs{1-e^{\phi(v_{j-1}-v_{j})}}
\abs{\sqrt{\pi_{\bl,\vec{\tau}_{1},\vec{\gamma}}(v_{0},\dots,v_{j-1},v_{j})}
-\sqrt{\pi_{\bl,\vec{\tau}_{2},\vec{\gamma}}(v_{0},\dots,v_{j-1},v_{j})}}^{2}.
\end{multlined}
\end{equation}
Inserting the inductive hypothesis into \eqref{eqn:contraction-properties_real-contraction_multiple-iteration_inductive-step} completes the proof.
\end{proof}

\begin{corollary}
\label{cor:contraction-properties_real-contraction_direct-estimate}
Under the same conditions as Proposition \ref{prop:contraction-properties_real-contraction_multiple-iteration}, defining
\begin{equation}
\label{cor:contraction-properties_real-contraction_Vk}
V_{k,v_{0}}(\vec{\gamma})
:=\int_{(\R^{d})^{k}}\dd{v}_{1}\dots\dd{v}_{k}
\prod_{\ell=1}^{k}\left[
\vec{\gamma}(v_{0},\dots,v_{\ell-1},v_{\ell-1})
\abs{1-e^{-\phi(v_{\ell-1}-v_{\ell})}}
\right],
\end{equation}
we have that
\begin{equation}
\label{eqn:contraction-properties_real-contraction_direct-estimate}
\begin{multlined}
\abs{\sqrt{\pi_{\bl,\vec{\tau}_{1},\vec{\gamma}}(v_{0})}-\sqrt{\pi_{\bl,\vec{\tau}_{2},\vec{\gamma}}(v_{0})}}^{2}
\le \tilde{\lambda}^{k} M(\tilde{z})^{k} V_{k,v_{0}}(\vec{\gamma})\norm{\vec{\tau}_{1}-\vec{\tau}_{2}}_{\infty}.
\end{multlined}
\end{equation}
\end{corollary}

\begin{proof}
Setting $j=1$ in \eqref{eqn:contraction-properties_real-contraction_multiple-iteration}, we get that
\begin{equation}
\begin{multlined}
\abs{\sqrt{\pi_{\bl,\vec{\tau}_{1},\vec{\gamma}}(v_{0})}
-\sqrt{\pi_{\bl,\vec{\tau}_{2},\vec{\gamma}}(v_{0})}}^{2}
\le M(\tilde{z})^{k}\int_{(\R^{d})^{k}}\dd{v}_{1}\dots\dd{v}_{k}
\\
\times \prod_{\ell=1}^{k}\left[\bl(v_{\ell-1})\vec{\gamma}(v_{0},\dots,v_{\ell-1},v_{\ell-1})\vphantom{\abs{1-e^{-\phi(v_{\ell-1}-v_{\ell})}}}\abs{1-e^{-\phi(v_{\ell-1}-v_{\ell})}}\right]
\abs{\sqrt{\vec{\tau}_{1}(v_{0},\dots,v_{k})}-\sqrt{\vec{\tau}_{2}(v_{0},\dots,v_{k})}}^{2}.
\end{multlined}
\end{equation}
Using the elementary bound $\abs{\sqrt{x}-\sqrt{y}}^{2}\le\abs{x-y}$ which holds for all $x,y\ge 0$, we bound
\begin{equation}
\abs{\sqrt{\vec{\tau}_{1}(v_{0},\dots,v_{k})}-\sqrt{\vec{\tau}_{2}(v_{0},\dots,v_{k})}}^{2}
\le\norm{\vec{\tau}_{1}-\vec{\tau}_{2}}_{\infty}.
\end{equation}
Then, applying the assumption that $\vec{\gamma}(v_{0},\dots,v_{j-1},\cdot)\bl(\cdot)$ takes values in $[0,\tl]$ for all $1\le j\le k$ completes the proof.
\end{proof}

So far, our computation has been general and does not depend on the specifics of the tree recursion.
To connect it back to the family of tree recursions studied in \S\ref{sec:tree-recursions}, we first define a counterpart of Michelen and Perkins' potential-weighted connective constant \cite[Definition 2]{michelen2023analyticity} for locally stable hard-core models as follows.

\begin{definition}
\label{def:tree-recursions_pwcc}
Recall the reference modulating function $\vec{\gamma}_{c}$ introduced in Definition \ref{def:tree-recursions_canonical-modulating-function}.
We define a $\phi$-associated constant
\begin{equation}
\label{eqn:tree-recursions_pwcc}
\pwcc:=\inf_{k\ge 1}\set{V_{k,\vb{0}}(\vec{\gamma}_{c})^{1/k}\mid V_{k,\vb{0}}(\vec{\gamma}_{c})\ge\aconst^{k}},
\end{equation}
where we used the translation invariance of the pair potential $\phi$ to anchor $v_{0}$ at the origin $\vb{0}$ for concreteness.
\end{definition}

\begin{remark}
We will discuss the reason for imposing the lower bound of $\aconst^{k}$ in the definition of $\pwcc$ in Remark \ref{rmk:contraction-properties_real-contraction_optimization-problem}.
\end{remark}

We now apply Corollary \ref{cor:contraction-properties_real-contraction_direct-estimate} to the tree recursions constructed in Theorem \ref{thm:tree-recursions_density-correspondence}, assuming that the boundary conditions there all take values in the same interval $[0,\tz^{2}]$.
For all $k\ge 1$ and depth-$k$ tree recursions $\pi_{\bl^{m},\vec{\tau},\vec{\gamma}}$, we have that 
\begin{equation}
\label{eqn:contraction-properties_real-contraction_Vk-domination}
V_{k,v_{0}}(\vec{\gamma})\le V_{k,v_{0}}(\vec{\gamma}_{c})
\end{equation}
since $0\le\vec{\gamma}\le\vec{\gamma}_{c}$ (see \eqref{eqn:tree-recursions_modulating-function-choice}).
Hence, in view of \eqref{eqn:contraction-properties_real-contraction_direct-estimate}, a natural characterization of when the tree recursions should constitute a contraction is if $\tl M(\tz)V_{k,v_{0}}(\vec{\gamma}_{c})<1$ for some $k\ge 1$, which is in turn guaranteed by the \emph{contraction condition}:
\begin{equation}
\label{eqn:contraction-properties_real-contraction_contraction-condition}
\tl M(\tz)\pwcc<1.
\end{equation}

We observe the following simple upper bound on $\pwcc$.

\begin{lemma}
\label{lem:contraction-properties_real-contraction_pwcc-upper-bound}
$\pwcc\le C_{\phi}$.
\end{lemma}

\begin{proof}
Observe that $\vec{\gamma}_{c}(\vb{0},\cdot)\equiv 1$ as $G^{(1)}=\R^{d}$ and both products in \eqref{eqn:tree-recursions_canonical-modulating-function-good-sequence} are empty, so 
\begin{equation}
V_{1,\vb{0}}(\vec{\gamma}_{c})=\int_{\R^{d}}\dd{v}_{1}\vec{\gamma}_{c}(\vb{0},\vb{0})\abs{1-e^{-\phi(\vb{0}-v_{1})}}=C_{\phi}\ge\aconst.
\end{equation}
Thus, $\pwcc\le V_{1,\vb{0}}(\vec{\gamma}_{c})=C_{\phi}$.
\end{proof}

\subsection{An optimization problem}
\label{sec:contraction-properties_optimization-problem}

Here, we study an optimization problem for $\tl$ with the self-map condition \eqref{eqn:contraction-properties_real-contraction_self-map-condition} and the contraction condition \eqref{eqn:contraction-properties_real-contraction_contraction-condition} as its constraints.
The main result of this subsection is as follows.

\begin{proposition}
\label{prop:contraction-properties_optimization-problem}
Consider the region in $\R^{2}_{\ge 0}$ consisting of the points $(\tilde{\lambda},\tilde{z})$ satisfying the following inequalities:
\begin{equation}
\label{eqn:contraction-properties_optimization-problem}
\begin{cases}
\tl e^{\aconst\tz^{2}}\le\tz^{2} \\
\tl M(\tz)\pwcc\le 1
\end{cases}.
\end{equation}
The points in this region with the largest $\tilde{\lambda}$-value are given by
\begin{equation}
\label{eqn:contraction-properties_optimization-problem_repulsive-optimizer}
\begin{cases}
\tilde{\lambda}=e/\pwcc
\\
\tilde{z}\ge \sqrt{e/\pwcc}
\end{cases}
\end{equation}
if $\aconst=0$ (i.e., the pair potential is purely repulsive), or by
\begin{equation}
\label{eqn:contraction-properties_optimization-problem_locally-stable-optimizer}
\begin{cases}
\tilde{\lambda}=e^{1-2W(e\aconst/\pwcc)}/\pwcc
\\
\tilde{z}=\sqrt{W(e\aconst/\pwcc)/\aconst}
\end{cases}
\end{equation}
if $\aconst>0$, where $W(\cdot)$ is the Lambert $W$-function.
\end{proposition}

Before proving Proposition \ref{prop:contraction-properties_optimization-problem}, we need an explicit expression for $M(\tz)$.

\begin{lemma}
\label{eqn:contraction-properties_optimization-problem_M}
Recall the definition of $M(\tz)$ in Lemma \ref{lem:contraction-properties_real-contraction_single-iteration}.
Explicitly,
\begin{equation}
\label{eqn:contraction-properties_optimization-problem_M_repulsive}
M(\tilde{z})=
\begin{cases}
e^{-\pconst\tilde{z}^{2}}\pconst\tilde{z}^{2} & \text{if }\pconst\tilde{z}^{2}\le 1
\\
e^{-1} & \text{else}
\end{cases}
\end{equation}
if $\aconst=0$, or 
\begin{equation}
\label{eqn:contraction-properties_optimization-problem_M_locally-stable}
M(\tilde{z})=
\begin{cases}
C_{\phi}\tz^{2}e^{-(\pconst-\aconst)\tz^{2}} & \text{if }C_{\phi}\tilde{z}^{2}\le 1
\\
e^{2\aconst\tz^{2}-1} & \text{else}
\end{cases}.
\end{equation}
if $\aconst>0$.
\end{lemma}

\begin{proof}
Denote the function defined in \eqref{eqn:contraction-properties_real-contraction_single-iteration_coefficient-function} by $f(z)$.
If $\aconst=0$, then $f(z)=e^{-\pconst z^{2}}\pconst z^{2}$, and \eqref{eqn:contraction-properties_optimization-problem_M_repulsive} is clear.
Otherwise, setting $f'(z)=0$ yields that the only critical points of $f(z)$ on $\R_{\ge 0}$ are $z=0$ and the positive solution $z=z_{c}$ to $\pconst z^{2}+\aconst\tilde{z}^{2}=1$, the latter of which is in the interval $[0,\tilde{z}]$ if and only if $C_{\phi}\tilde{z}^{2}\ge 1$.
Notice that if $C_{\phi}\tilde{z}^{2}\ge 1$, then
\begin{equation}
\frac{f(z_{c})}{f(0)}=\frac{e^{\aconst\tz^{2}-1}}{\aconst\tz^{2}}\ge 1
\text{ and }
\frac{f(z_{c})}{f(\tz)}=\frac{e^{C_{\phi}\tz^{2}-1}}{C_{\phi}\tz^{2}}\ge 1,
\end{equation}
so $M(\tilde{z})=f(z_{c})$; otherwise, since
\begin{equation}
\frac{f(0)}{f(\tz)}=\frac{\aconst}{C_{\phi}}e^{\pconst\tz^{2}}\le\left(1-\frac{\pconst}{C_{\phi}}\right)e^{\pconst/C_{\phi}}\le 1,
\end{equation}
we have that $M(\tilde{z})=f(\tilde{z})$.
\end{proof}

We now prove Proposition \ref{prop:contraction-properties_optimization-problem}.

\begin{proof}[Proof of \eqref{eqn:contraction-properties_optimization-problem_repulsive-optimizer}]
Suppose $\aconst=0$.
We will compare the bounds on $\tilde{\lambda}$ given by the two inequalities in \eqref{eqn:contraction-properties_optimization-problem}.

Suppose first that $\pconst\tz^{2}\le 1$.
By Lemma \ref{lem:contraction-properties_real-contraction_pwcc-upper-bound}, $\pwcc\le C_{\phi}=\pconst$, so
\begin{equation}
M(\tz)\pwcc\tz^{2}
\le(\pconst\tz^{2})^{2}e^{-\pconst\tz^{2}}\le e^{-1},
\end{equation}
that is, the self-map condition always gives a tighter bound.
Thus, $\tilde{\lambda}\le\tz^{2}\le\pconst^{-1}$ in this case.

Suppose now that $\pconst\tz^{2}\ge 1$.
The contraction condition gives the $\tz$-independent bound of $\tl\le e/\pwcc$, which is attained if $\tz^{2}\ge e/\pwcc$, but the latter is ensured by the assumption that $\pconst\tz^{2}\ge 1$.
The proof of \eqref{eqn:contraction-properties_optimization-problem_repulsive-optimizer} is now complete.
\end{proof}

\begin{proof}[Proof of \eqref{eqn:contraction-properties_optimization-problem_locally-stable-optimizer}]
Suppose that $\aconst>0$.

Suppose first that $C_{\phi}\tz^{2}\le 1$.
Since 
\begin{equation}
M(\tz)\pwcc e^{-\aconst\tz^{2}}\tz^{2}
=C_{\phi}\pwcc\tz^{4}e^{-\pconst\tz^{2}}
\le(C_{\phi}\tz^{2})^{2}e^{-\pconst\tz^{2}} \le 1,
\end{equation}
the self-map condition always gives a tighter bound.
We conclude that, in this case, $\tl\le e^{-\aconst/C_{\phi}}/C_{\phi}$.

Suppose now that $C_{\phi}\tz^{2}\ge 1$.
Setting the upper bounds on $\tl$ from the self-map and contraction conditions equal ($\frac{1}{\pwcc}e^{1-2\aconst\tz^{2}}=\tz^{2}e^{-\aconst\tz^{2}}$), we find that $\aconst\tz^{2}=W(e\aconst/\pwcc)$.
Note that this solution satisfies our assumption that $C_{\phi}\tz^{2}\ge 1$, since
\begin{equation}
\frac{C_{\phi}}{\aconst}W(eA_{\phi}/\pwcc)\ge \frac{C_{\phi}}{\aconst}W(eA_{\phi}/C_{\phi})\ge 1.
\end{equation}
Moreover, it gives $\tl=e^{1-2W(e\aconst/\pwcc)}/\pwcc$, which is larger than the bound on $\tl$ in the previous case, since
\begin{equation}
\label{eqn:contraction-properties_real-contraction_optimization-problem_assumption-used-1}
\frac{e^{1-2W(e\aconst/\pwcc)}/\pwcc}{e^{-\aconst/C_{\phi}}/C_{\phi}}
=\frac{e^{1-2W(e\aconst/\pwcc)}\aconst/\pwcc}{e^{1-2W(e\aconst/C_{\phi})}\aconst/C_{\phi}}
\cdot \frac{e^{1-2W(e\aconst/C_{\phi})}}{e^{-\aconst/C_{\phi}}}
\ge 1.
\end{equation}

To complete the proof of \eqref{eqn:contraction-properties_optimization-problem_locally-stable-optimizer}, we need to check that no other value of $\tz$ would allow for a larger $\tl$.
Since the upper bound for $\tl$ from the contraction condition is a decreasing function of $\tz$, it suffices to study the $\tz$ such that $\aconst\tz^{2}<W(e\aconst/\pwcc)$.
Comparing the two upper bounds on $\tl$, we find that
\begin{equation}
\frac{\tz^{2}e^{-\aconst\tz^{2}}}{\frac{1}{\pwcc}e^{1-2\aconst\tz^{2}}}
=\pwcc\tz^{2}e^{\aconst\tz^{2}-1}
<\frac{\pwcc}{\aconst}W(e\aconst/\pwcc)e^{W(e\aconst/\pwcc)-1}=1,
\end{equation}
so for this range of $\tz$, the self-map condition always gives a tighter bound.
However, for 
\begin{equation}
\label{eqn:contraction-properties_real-contraction_optimization-problem_assumption-used-2}
\aconst\tz^{2}<W(e\aconst/\pwcc)\le1,
\end{equation}
the upper bound for $\tl$ from the self-map condition is an increasing function of $\tz$, so smaller values of $\tz$ would lead to smaller values of $\tl$ than before.
\end{proof}

\begin{remark}
\label{rmk:contraction-properties_real-contraction_optimization-problem}
The only places where we used the relation $\pwcc\ge\aconst$ are in \eqref{eqn:contraction-properties_real-contraction_optimization-problem_assumption-used-1} and \eqref{eqn:contraction-properties_real-contraction_optimization-problem_assumption-used-2}.
As the above computations demonstrate, without this relation, we would not be able to conclude that the choice of $\tl$ in \eqref{eqn:contraction-properties_optimization-problem_locally-stable-optimizer} is optimal, and additional analysis will be needed.
This may be the case if we had instead defined $\pwcc:=\inf_{k\ge 1} V_{k}(\vec{\gamma})^{1/k}$: it is not obvious if the relation $\inf_{k\ge 1} V_{k}(\vec{\gamma})^{1/k}\ge\aconst$ should hold, at least by a perfunctory examination of \eqref{cor:contraction-properties_real-contraction_Vk}.
\end{remark}

\subsection{Complex contraction}

We now turn to the more delicate case of complex-valued activity and one-point density functions.
Our computation in \S\ref{sec:contraction-properties_real-contraction} is not good enough for this purpose because the potential function used there ($\psi(x)=\sqrt{x}$) is not analytic at $0$.
Following the treatment in \cite[\S A.2]{michelen2023analyticity}, we introduce a tiny shift $\delta>0$ (to be chosen in \S\ref{sec:contraction-properties_complex-contraction_choices-of-constants}) to the previous potential function and study the transformation of the functional \eqref{eqn:contraction-properties_functional} under the new potential function $\psi(z)=\sqrt{\delta+z}$, which is indeed analytic at $0$:
\begin{equation}
\begin{multlined}
h_{v}(\bl,\bz)
:=\psi(F_{v}(\bl,\psi^{-1}(\bz)))
\\
=\left[\delta+\bl(v)\exp{-\int_{\R^{d}}\dd{w}\left[1-e^{-\phi(v-w)}\right][\bz(w)^{2}-\delta]}\right]^{1/2}.
\end{multlined}
\end{equation}
We will specify the domain of definition of the transformed functional $h_{v}$ in \S\ref{sec:contraction-properties_complex-contraction_choices-of-constants}.
As we will see, many properties of the old functional $g_{v}(\bl,\bz)$ as computed in \S\ref{sec:contraction-properties_real-contraction} transfer to the new one with only minor modifications, provided that $\delta>0$ is sufficiently small.

\begin{remark}
In terms of analyzing the contraction properties of the tree recursions, most of our new ideas have already been introduced in \S\ref{sec:contraction-properties_real-contraction} and \S\ref{sec:contraction-properties_optimization-problem}.
As such, this subsection does not contain much substantial innovation compared to \cite[\S A.2]{michelen2023analyticity}, apart from our attempt to make various aspects of the argument more explicit.
Hence, readers familiar with the present topic may glance through the statements of Proposition \ref{prop:contraction-properties_complex-contraction_multiple-iteration} and Corollaries \ref{cor:contraction-properties_complex-contraction_direct-estimate} and \ref{cor:contraction-properties_complex-contraction_final-result} below and skip directly to \S\ref{sec:zero-freeness}.
\end{remark}

\subsubsection{Choices of constants}
\label{sec:contraction-properties_complex-contraction_choices-of-constants}

First, we select a number of constants that will be used repeatedly in this subsection.
Let $(\tl,\tz)\in\R_{>0}^{2}$ be a solution to the optimization problem in Proposition \ref{prop:contraction-properties_optimization-problem}, where we choose any $\tz\ge\sqrt{e/\pwcc}$ if $\aconst=0$.
Fix $\lambda\in(0,\tl)$.
Let $\epsilon>0$ be such that $\lambda\le e^{-5\epsilon}\tl$.
By Definition \ref{def:tree-recursions_pwcc}, there exists $k\ge 1$ such that $V_{k,\vb{0}}(\vec{\gamma}_{c})^{1/k}\le e^{\epsilon}\pwcc$.

For any $[a,b]\subset\R$ and $r>0$, we denote the (open) complex neighborhood of $[a,b]$ of radius $r$ by
\begin{equation}
\mathcal{N}_{r}([a,b]):=\set{z\in\C\mid\text{there exists }x\in[a,b]\text{ such that }\abs{z-x}<r}.
\end{equation}
To ensure that the functional $h_{v}(\bl,\bz)$ is well-defined, we will restrict ourselves to those $\bl$ and $\bz$ that take values in $\mathcal{N}_{\delta_{1}}([0,\lambda])$ and $\mathcal{N}_{\delta_{2}}([\sqrt{\delta},\sqrt{\delta+\tz^{2}}])$, respectively, where $\delta_{1},\delta_{2}>0$ are very small constants that we choose below.

First, we choose $\delta\in(0,\frac{\epsilon}{3C_{\phi}}]$ such that
\begin{enumerate}[label=(\roman*)]
\item the maximum of the function 
\begin{equation}
(z_{1},z_{2})\mapsto(\pconst z_{1}^{2}+\aconst z_{2}^{2})e^{-\pconst z_{1}^{2}+\aconst z_{2}^{2}}
\end{equation}
on $[\sqrt{\delta},\sqrt{\delta+\tz^{2}}]^{2}$ does not exceed $e^{\epsilon/6}M(\tz)$ (compare with Lemma \ref{lem:contraction-properties_real-contraction_single-iteration}).
\label{itm:contraction-properties_complex-contraction_choices-of-constants_shifted-real-maximum}
\end{enumerate}

Next, we choose
\begin{enumerate}[label=(\roman*)]
\setcounter{enumi}{1}
\item $b_{1}>0$ and $c_{1}\in(0,\tl(e^{\epsilon/2}-1)]$ such that the auxiliary function 
\begin{equation}
h(\lambda,z_{1},z_{2}):=\sqrt{\delta+\lambda \exp{-\pconst (z_{1}^{2}-\delta)+\aconst (z_{2}^{2}-\delta)}}.
\end{equation}
is well-defined on $\mathcal{N}_{b_{1}}([0,\lambda])\times\mathcal{N}_{c_{1}}([\sqrt{\delta},\sqrt{\delta+\tz^{2}}])^{2}$;
\label{itm:contraction-properties_complex-contraction_choices-of-constants_well-definedness}
\item $c_{2}\in(0,c_{1}]$ such that the maximum of the function
\begin{equation}
(z_{1},z_{2})\mapsto(\pconst \abs{z_{1}}^{2}+\aconst \abs{z_{2}}^{2})e^{-\pconst \abs{z_{1}}^{2}+\aconst \abs{z_{2}}^{2}}
\end{equation}
on $\mathcal{N}_{c_{1}}([\sqrt{\delta},\sqrt{\delta+\tz^{2}}])^{2}$ does not exceed $e^{\epsilon/3}M(\tz)$ (as an upgrade of Property \ref{itm:contraction-properties_complex-contraction_choices-of-constants_shifted-real-maximum});
\label{itm:contraction-properties_complex-contraction_choices-of-constants_shifted-complex-maximum}
\item $b_{2}\in(0,b_{1}]$ and $c_{3}\in(0,c_{1}]$ such that
\begin{equation}
\abs{\delta+\lambda \exp{-\pconst (z_{1}^{2}-\delta)+\aconst (z_{2}^{2}-\delta)}}
\ge\abs{\lambda\exp{-\pconst(z_{1}^{2}-\delta)+\aconst(z_{2}^{2}-\delta)}}
\end{equation}
on $\mathcal{N}_{b_{2}}([0,\lambda])\times\mathcal{N}_{c_{3}}([\sqrt{\delta},\sqrt{\delta+\tz^{2}}])^{2}$;
\label{itm:contraction-properties_complex-contraction_choices-of-constants_remove-shift}
\item $c_{4}\in(0,c_{1}]$ such that 
\begin{equation}
\abs{z^{2}-\abs{z}^{2}}\le\frac{\epsilon}{3C_{\phi}}
\end{equation}
on $\mathcal{N}_{c_{4}}([\sqrt{\delta},\sqrt{\delta+\tz^{2}}])$;
\label{itm:contraction-properties_complex-contraction_choices-of-constants_square}
\item $b_{3}>0$ such that $\abs{\lambda'}\le e^{\epsilon}\lambda$ for all $\lambda'\in\mathcal{N}_{b_{3}}([0,\lambda])$.
\label{itm:contraction-properties_complex-contraction_choices-of-constants_activity-modulus}
\end{enumerate}

We will use convexity \cite{peters2019conjecture} to extend Properties \ref{itm:contraction-properties_complex-contraction_choices-of-constants_well-definedness} and \ref{itm:contraction-properties_complex-contraction_choices-of-constants_remove-shift} to the case of complex-valued functions.

\begin{lemma}
\label{lem:contraction-properties_complex-contraction_convex-image}
If $c_{5}\in(0,\sqrt{\delta}]$, then the image of $U=\mathcal{N}_{c_{5}}([\sqrt{\delta},\sqrt{\delta+\tz^{2}}])$ under the map $z\mapsto z^{2}$ is convex.
\end{lemma}

\begin{proof}
Let $c_{5}\in(0,\sqrt{\delta}]$.
Since $U$ is bounded and the function $f(z)=z^{2}$ is conformal and injective on $U$, the boundary of the image of $U$ coincides with the image of the boundary of $U$, i.e., $\partial f(U)=f(\partial U)$.
In particular, $f(\partial U)$ is a $C^{1}$ and piecewise $C^{2}$ simple closed curve in $\C$, composed of four $C^{2}$ pieces corresponding respectively to the two straight edges and two circular arcs in $\partial U$.
Hence, to prove that $f(U)$ is convex, it suffices to check that the signed curvature of $f(\partial U)$, wherever defined, has constant sign.\footnote{The author thanks Guanhua Shao for discussions about this point.}
The latter follows from an elementary computation.
\end{proof}

We set $\delta_{1}=\min\set{b_{1},b_{2},b_{3}}$, $\delta_{2}'=\min\set{c_{1},c_{2},c_{3},c_{4},\sqrt{\delta}}$, $B=\max\set{(\tl M(\tz)C_{\phi})^{1/2},1}$, and $\delta_{2}=B^{-k/2}\delta_{2}'$.
By Lemma \ref{lem:contraction-properties_complex-contraction_convex-image}, for any measurable function $\bz:\R^{d}\rightarrow\mathcal{N}_{\delta_{2}'}([\sqrt{\delta},\sqrt{\delta+\tz^{2}}])$, we can find $z_{1},z_{2}\in\mathcal{N}_{\delta_{2}'}([\sqrt{\delta},\sqrt{\delta+\tz^{2}}])$ as in \eqref{eqn:contraction-properties_real-contraction_single-iteration_convexity} (recall the region $\mathcal{D}_{v}$ defined there) such that
\begin{equation}
\pconst z_{1}^{2}=\int_{\mathcal{D}_{v}}\dd{w}\left[1-e^{-\phi(v-w)}\right]\bz(w)^{2}, \quad
\aconst z_{2}^{2}=\int_{\R^{d}\setminus\mathcal{D}_{v}}\dd{w}\left[e^{-\phi(v-w)}-1\right]\bz(w)^{2}.
\end{equation}
Hence, we have the following extensions of Properties \ref{itm:contraction-properties_complex-contraction_choices-of-constants_well-definedness} and \ref{itm:contraction-properties_complex-contraction_choices-of-constants_remove-shift}:
\begin{enumerate}[label=(\roman*')]
\setcounter{enumi}{1}
\item the functional $h_{v}(\bl,\bz)$ is well-defined for $\bl$ and $\bz$ respectively taking values in $\mathcal{N}_{\delta_{1}}([0,\lambda])$ and $\mathcal{N}_{\delta_{2}'}([\sqrt{\delta},\sqrt{\delta+\tz^{2}}])$, since
\begin{equation}
\begin{multlined}
h_{v}(\bl,\bz)=\left[\delta+\bl(v)\exp{-\int_{\R^{d}}\dd{w}\left[1-e^{-\phi(v-w)}\right][\bz(w)^{2}-\delta]}\right]^{1/2}
\\
=\sqrt{\delta+\bl(v)\exp{-\pconst(z_{1}^{2}-\delta)}+\aconst(z_{2}^{2}-\delta)}
=h(\bl(v),z_{1},z_{2});
\end{multlined}
\end{equation}
\label{itm:contraction-properties_complex-contraction_choices-of-constants_function-well-definedness}
\setcounter{enumi}{3}
\item for $\bl$ and $\bz$ respectively taking values in $\mathcal{N}_{\delta_{1}}([0,\lambda])$ and $\mathcal{N}_{\delta_{2}'}([\sqrt{\delta},\sqrt{\delta+\tz^{2}}])$, it holds that
\begin{equation}
\begin{multlined}
\abs{\delta+\bl(v)\exp{-\int_{\R^{d}}\dd{w}\left[1-e^{-\phi(v-w)}\right]\bz(w)^{2}}}
\\
\ge\abs{\bl(v)\exp{-\int_{\R^{d}}\dd{w}\left[1-e^{-\phi(v-w)}\right]\bz(w)^{2}}}.
\end{multlined}
\end{equation}
\label{itm:contraction-properties_complex-contraction_choices-of-constants_function-remove-shift}
\end{enumerate}

\subsubsection{Complex contraction}

From here on, we consider the functional $h_{v}(\bl,\bz)$ as defined over measurable functions $\bl:\R^{d}\rightarrow\mathcal{N}_{\delta_{1}}([0,\lambda])$ and $\bz:\R^{d}\rightarrow\mathcal{N}_{\delta_{2}}([\sqrt{\delta},\sqrt{\delta+\tz^{2}}])$.
As in \S\ref{sec:contraction-properties_real-contraction}, we will first analyze the functional and then apply it to tree recursions.

We have the following analogue of Lemma \ref{lem:contraction-properties_real-contraction_single-iteration}, where we estimate the difference in the values of the functional when applied to the same $\bl$ but different $\bz$'s.

\begin{lemma}
\label{lem:contraction-properties_complex-contraction_single-iteration}
For any Lebesgue measurable functions $\bl:\R^{d}\rightarrow\mathcal{N}_{\delta_{1}}([0,\lambda])$ and $\bx,\by:\R^{d}\rightarrow\mathcal{N}_{\delta_{2}'}([\sqrt{\delta},\sqrt{\delta+\tz^{2}}])$ (noting that the radius here is $\delta_{2}'$ rather than $\delta_{2}$),
\begin{equation}
\label{eqn:contraction-properties_complex-contraction_single-iteration}
\abs{h_{v}(\bl,\bx)-h_{v}(\bl,\by)}^{2}
\le e^{\epsilon}\abs{\bl(v)}M(\tilde{z})\int_{\R^{d}}\dd{w}\abs{1-e^{-\phi(v-w)}}\abs{\bx(w)-\by(w)}^{2}.
\end{equation}
\end{lemma}

\begin{proof}
Let $\bz_{s}:=s\bx+(1-s)\by$.
By the mean value theorem, there exists $t\in(0,1)$ such that
\begin{equation}
\label{eqn:contraction-properties_complex-contraction_single-iteration_derivative}
\begin{multlined}
\abs{h_{v}(\bl,\bx)-h_{v}(\bl,\by)}
\le\abs{\eval{\dv{s}h_{v}(\bl,\bz_{s})}_{s=t}}
\\
\le \frac{\abs{\bl(v)}\abs{\exp{-\int_{\R^{d}}\dd{w}\left[1-e^{-\phi(v-w)}\right]\left[\bz_{t}(w)^{2}-\delta\right]}}}{\abs{\delta+\bl(v)\exp{-\int_{\R^{d}}\dd{w}\left[1-e^{-\phi(v-w)}\right][\bz_{t}(w)^{2}-\delta]}}^{1/2}}
\\
\times
\abs{\int_{\R^{d}}\dd{w}\left[1-e^{-\phi(v-w)}\right]\bz_{t}(w)\left[\by(w)-\bx(w)\right]}.
\end{multlined}
\end{equation}
As in \eqref{eqn:contraction-properties_real-contraction_single-iteration_intermediate-bound}, we bound the last term of \eqref{eqn:contraction-properties_complex-contraction_single-iteration_derivative} using the Cauchy-Schwarz inequality with the measure $\dd{w}\abs{1-e^{-\phi(v-w)}}$ on $\R^{d}$:
\begin{equation}
\begin{multlined}
\abs{\int_{\R^{d}}\dd{w}\left[1-e^{-\phi(v-w)}\right]\bz_{t}(w)\left[\by(w)-\bx(w)\right]}
\\
\le \left(\int_{\R^{d}}\dd{w}\abs{1-e^{-\phi(v-w)}}\abs{\bz_{t}(w)}^{2}
\right)^{1/2}
\left(\int_{\R^{d}}\dd{w}\abs{1-e^{-\phi(v-w)}}\abs{\by(w)-\bx(w)}^{2}\right)^{1/2}.
\end{multlined}
\end{equation}
As for the ratio in \eqref{eqn:contraction-properties_complex-contraction_single-iteration_derivative}, we first use Property \ref{itm:contraction-properties_complex-contraction_choices-of-constants_function-remove-shift} to remove the denominator:
\begin{equation}
\begin{multlined}
\frac{\abs{\bl(v)}\abs{\exp{-\int_{\R^{d}}\dd{w}\left[1-e^{-\phi(v-w)}\right]\left[\bz_{t}(w)^{2}-\delta\right]}}}{\abs{\delta+\bl(v)\exp{-\int_{\R^{d}}\dd{w}\left[1-e^{-\phi(v-w)}\right][\bz_{t}(w)^{2}-\delta]}}^{1/2}}
\\
\le\abs{\bl(v)}^{1/2}\abs{\exp{-\int_{\R^{d}}\dd{w}\left[1-e^{-\phi(v-w)}\right]\left[\bz_{t}(w)^{2}-\delta\right]}}^{1/2}.
\end{multlined}
\end{equation}
Next, we use Property \ref{itm:contraction-properties_complex-contraction_choices-of-constants_square} to bound the modulus of the exponential:
\begin{equation}
\begin{split}
{}&\abs{e^{-\int_{\R^{d}}\dd{w}\left[1-e^{-\phi(v-w)}\right]\left[\bz_{t}(w)^{2}-\delta\right]}} \\
={}&e^{-\int_{\R^{d}}\dd{w}\left[1-e^{-\phi(v-w)}\right]\left[\abs{\bz_{t}(w)}^{2}-\delta\right]}
\abs{e^{-\int_{\R^{d}}\dd{w}\left[1-e^{-\phi(v-w)}\right]\left[\bz_{t}(w)^{2}-\abs{\bz_{t}(w)}^{2}\right]}} \\
\le{}&e^{-\int_{\R^{d}}\dd{w}\left[1-e^{-\phi(v-w)}\right]\left[\abs{\bz_{t}(w)}^{2}-\delta\right]}
e^{\int_{\R^{d}}\dd{w}\abs{1-e^{-\phi(v-w)}}\abs{\bz_{t}(w)^{2}-\abs{\bz_{t}(w)}^{2}}} \\
\le{}& e^{2\epsilon/3} e^{-\int_{\R^{d}}\dd{w}\left[1-e^{-\phi(v-w)}\right]\abs{\bz_{t}(w)}^{2}},
\end{split}
\end{equation}
where we recalled also that $\delta\le\frac{\epsilon}{3C_{\phi}}$ in the last inequality.
Combining all of the above, we get that
\begin{equation}
\begin{multlined}
\abs{h_{v}(\bl,\bx)-h_{v}(\bl,\by)}^{2}
\le e^{2\epsilon/3}\abs{\bl(v)}
e^{-\int_{\R^{d}}\dd{w}\left[1-e^{-\phi(v-w)}\right]\abs{\bz_{t}(w)}^{2}}
\\
\times
\int_{\R^{d}}\dd{w}\abs{1-e^{-\phi(v-w)}}\abs{\bz_{t}(w)}^{2}
\cdot\int_{\R^{d}}\dd{w}\abs{1-e^{-\phi(v-w)}}\abs{\by(w)-\bx(w)}^{2}.
\end{multlined}
\end{equation}
Finally, bounding
\begin{equation}
e^{-\int_{\R^{d}}\dd{w}\left[1-e^{-\phi(v-w)}\right]\abs{\bz_{t}(w)}^{2}}
\int_{\R^{d}}\dd{w}\abs{1-e^{-\phi(v-w)}}\abs{\bz_{t}(w)}^{2}
\le e^{\epsilon/3}M(\tz)
\end{equation}
using Property \ref{itm:contraction-properties_complex-contraction_choices-of-constants_shifted-complex-maximum} completes the proof.
\end{proof}

The idea now is to iterate the bound in Lemma \ref{lem:contraction-properties_complex-contraction_single-iteration} for depth-$k$ tree recursions and derive an analogue of Proposition \ref{prop:contraction-properties_real-contraction_multiple-iteration}, where $k\ge 1$ is fixed as in \S\ref{sec:contraction-properties_complex-contraction_choices-of-constants}.
However, in contrast to the real-valued case, we have here the question of whether we can actually iterate the bound at each step in the tree recursion, as we did in the proof of Proposition \ref{prop:contraction-properties_real-contraction_multiple-iteration}.
Indeed, observe that the restriction of the functional $h_{v}(\bl,\bz)$ to functions $\bl$ and $\bz$ respectively taking values in the real intervals $[0,\lambda]$ and $[\sqrt{\delta},\sqrt{\delta+\tz^{2}}]$ satisfies an analogue of the self-map condition \eqref{eqn:contraction-properties_real-contraction_self-map-condition}: 
\begin{equation}
\label{eqn:contraction-properties_complex-contraction_self-map-condition}
h_{v}(\bl,\bz)\in[\sqrt{\delta},\sqrt{\delta+\tz^{2}}].
\end{equation}
With complex-valued functions, however, it is not obvious if a similar property should hold.
In principle, each iteration can take us farther away from the real interval $[\sqrt{\delta},\sqrt{\delta+\tz^{2}}]$ and possibly even out of the domain of definition of $h_{v}$.
Thus, our first order of business is to show that this cannot happen given our choices of the constants in \S\ref{sec:contraction-properties_complex-contraction_choices-of-constants}.

\begin{lemma}
\label{lem:contraction-properties_complex-contraction_deviation-extent}
Recall the constant $B>0$ from \S\ref{sec:contraction-properties_complex-contraction_choices-of-constants}.
For all $0\le j<k$ and measurable functions $\bl:\R^{d}\rightarrow\mathcal{N}_{\delta_{1}}([0,\lambda])$ and $\bz:\R^{d}\rightarrow\mathcal{N}_{B^{j}\delta_{2}}([\sqrt{\delta},\sqrt{\delta+\tz^{2}}])$, we have that
\begin{equation}
h_{v}(\bl,\bz)\in\mathcal{N}_{B^{j+1}\delta_{2}}([\sqrt{\delta},\sqrt{\delta+\tz^{2}}]).
\end{equation}
\end{lemma}

\begin{proof}
Define $\pi_{1}(z):=\Re z$ for $z\in\C$ and 
\begin{equation}
\pi_{2}(x):=\sqrt{\delta}\indicator{(-\infty,\sqrt{\delta})}(x)+x\indicator{[\sqrt{\delta},\sqrt{\delta+\tz^{2}}]}(x)+\sqrt{\delta+\tz^{2}}\indicator{(\sqrt{\delta+\tz^{2}},\infty)}(x)
\end{equation}
for $x\in\R$.
Set $\bx:=\pi_{2}\circ\pi_{1}\circ\bz$.\footnote{The author thanks Anupam Nayak for this construction.}
Note that $\bx$ is measurable because $\pi_{1}$ and $\pi_{2}$ are both continuous and $\bz$ is measurable.
Moreover, $\bx$ takes values in $[\sqrt{\delta},\sqrt{\delta+\tz^{2}}]$, and $\norm{\bz-\bx}_{\infty}\le B^{j}\delta_{2}$.
Applying \eqref{eqn:contraction-properties_complex-contraction_single-iteration} to $\bz$ and $\bx$, we get that
\begin{equation}
\abs{h_{v}(\bl,\bz)-h_{v}(\bl,\bx)}
\le \left(e^{\epsilon}\abs{\bl(v)}M(\tilde{z})C_{\phi}\right)^{1/2}B^{j}\delta_{2}.
\end{equation}
We bound $\abs{\bl(v)}\le e^{\epsilon}\lambda\le e^{-4\epsilon}\tl$ using Property \ref{itm:contraction-properties_complex-contraction_choices-of-constants_activity-modulus}.
Recalling the self-map property \eqref{eqn:contraction-properties_complex-contraction_self-map-condition} and the definition of $B$ complete the proof.
\end{proof}

Using Lemma \ref{lem:contraction-properties_complex-contraction_deviation-extent}, we apply Lemma \ref{lem:contraction-properties_complex-contraction_single-iteration} to depth-$k$ tree recursions.

\begin{proposition}
\label{prop:contraction-properties_complex-contraction_multiple-iteration}
Suppose that $\pi_{\bl,\vec{\tau}_{1},\vec{\gamma}}$ and $\pi_{\bl,\vec{\tau}_{2},\vec{\gamma}}$ are two depth-$k$ tree recursions possibly differing in their boundary conditions, where $\vec{\gamma}(v_{0},\dots,v_{j-1},\cdot)\bl(\cdot)$ takes values in $\mathcal{N}_{\delta_{1}}([0,\lambda])$ for all $1\le j\le k$ and $v_{0},\dots,v_{j-1}\in\R^{d}$, and both boundary conditions $\vec{\tau}_{i}$ take values in $\psi^{-1}(\mathcal{N}_{\delta_{2}}([\sqrt{\delta},\sqrt{\delta+\tz^{2}}]))$.
Then, for all $1\le j\le k+1$ and $v_{0},\dots,v_{j-1}\in\R^{d}$,
\begin{equation}
\label{eqn:contraction-properties_complex-contraction_multiple-iteration}
\begin{multlined}
\abs{\psi(\pi_{\bl,\vec{\tau}_{1},\vec{\gamma}}(v_{0},\dots,v_{j-1}))
-\psi(\pi_{\bl,\vec{\tau}_{2},\vec{\gamma}}(v_{0},\dots,v_{j-1}))}^{2}
\le (e^{\epsilon}M(\tilde{z}))^{k+1-j}\int_{(\R^{d})^{k+1-j}}\dd{v}_{j}\dots\dd{v}_{k}
\\
\times \prod_{\ell=j}^{k}\left[\abs{\bl(v_{\ell-1})}\vec{\gamma}(v_{0},\dots,v_{\ell-1},v_{\ell-1})\vphantom{\abs{1-e^{-\phi(v_{\ell-1}-v_{\ell})}}}\abs{1-e^{-\phi(v_{\ell-1}-v_{\ell})}}\right]
\abs{\psi(\vec{\tau}_{1}(v_{0},\dots,v_{k}))-\psi(\vec{\tau}_{2}(v_{0},\dots,v_{k}))}^{2}.
\end{multlined}
\end{equation}
\end{proposition}

\begin{proof}
In addition to \eqref{eqn:contraction-properties_complex-contraction_multiple-iteration}, we will prove by induction on $k+1-j$ that, for all $1\le j\le k+1$ and $v_{0},\dots,v_{j-1}\in\R^{d}$,
\begin{equation}
\label{eqn:contraction-properties_complex-contraction_multiple-iteration_deviation-extent}
\psi(\pi_{\bl,\vec{\tau}_{i},\vec{\gamma}}(v_{0},\dots,v_{j-1}))\in\mathcal{N}_{B^{k+1-j}\delta_{2}}([\sqrt{\delta},\sqrt{\delta+\tz^{2}}]).
\end{equation}
When $j=k+1$, the tree recursion returns the boundary condition, and both \eqref{eqn:contraction-properties_complex-contraction_multiple-iteration} and \eqref{eqn:contraction-properties_complex-contraction_multiple-iteration_deviation-extent} are clear.
Suppose now that $j\le k$.
Carrying out one step in either tree recursion, we see that
\begin{equation}
\begin{split}
{}&\psi(\pi_{\bl,\vec{\tau}_{i},\vec{\gamma}}(v_{0},\dots,v_{j-1}))
\\
={}&
\begin{multlined}[t]
\left[\delta+\bl(v_{j-1})\vec{\gamma}(v_{0},\dots,v_{j-1},v_{j-1})
\vphantom{\times\exp{-\int_{\R^{d}}\dd{v}_{j}\left[1-e^{-\phi(v_{j-1}-v_{j})}\right]\pi_{\bl,\vec{\tau}_{i},\vec{\gamma}}(v_{0},\dots,v_{j-1},v_{j})}}
\right.
\\
\left.\times\exp{-\int_{\R^{d}}\dd{v}_{j}\left[1-e^{-\phi(v_{j-1}-v_{j})}\right]\pi_{\bl,\vec{\tau}_{i},\vec{\gamma}}(v_{0},\dots,v_{j-1},v_{j})}\right]^{1/2}
\end{multlined}
\\
={}&h_{v_{j-1}}(\bl(\cdot)\vec{\gamma}(v_{0},\dots,v_{j-1},\cdot),\psi(\pi_{\bl,\vec{\tau}_{i},\vec{\gamma}}(v_{0},\dots,v_{j-1},\cdot))).
\end{split}
\end{equation}

We prove \eqref{eqn:contraction-properties_complex-contraction_multiple-iteration_deviation-extent} first.
By the inductive hypothesis, 
\begin{equation}
\psi(\pi_{\bl,\vec{\tau}_{i},\vec{\gamma}}(v_{0},\dots,v_{j-1},\cdot))\in\mathcal{N}_{B^{k-j}\delta_{2}}([\sqrt{\delta},\sqrt{\delta+\tz^{2}}]).
\end{equation}
On the other hand, $\vec{\gamma}(v_{0},\dots,v_{j-1},\cdot)\bl(\cdot)\in\mathcal{N}_{\delta_{1}}([0,\lambda])$ by assumption.
Thus, 
\begin{equation}
\psi(\pi_{\bl,\vec{\tau}_{i},\vec{\gamma}}(v_{0},\dots,v_{j-1}))\in\mathcal{N}_{B^{k+1-j}\delta_{2}}([\sqrt{\delta},\sqrt{\delta+\tz^{2}}])
\end{equation}
by Lemma \ref{lem:contraction-properties_complex-contraction_deviation-extent}.

We now prove \eqref{eqn:contraction-properties_complex-contraction_multiple-iteration}.
By Lemma \ref{lem:contraction-properties_complex-contraction_single-iteration},
\begin{equation}
\label{eqn:contraction-properties_complex-contraction_multiple-iteration_inductive-step}
\begin{multlined}
\abs{\psi(\pi_{\bl,\vec{\tau}_{1},\vec{\gamma}}(v_{0},\dots,v_{j-1}))
-\psi(\pi_{\bl,\vec{\tau}_{2},\vec{\gamma}}(v_{0},\dots,v_{j-1}))}^{2}
\le e^{\epsilon}\abs{\bl(v_{j-1})}
\vec{\gamma}(v_{0},\dots,v_{j-1},v_{j-1})
M(\tz)
\\
\times
\int_{\R^{d}}\dd{v}_{j}\abs{1-e^{\phi(v_{j-1}-v_{j})}}
\abs{\psi(\pi_{\bl,\vec{\tau}_{1},\vec{\gamma}}(v_{0},\dots,v_{j-1},v_{j}))
-\psi(\pi_{\bl,\vec{\tau}_{2},\vec{\gamma}}(v_{0},\dots,v_{j-1},v_{j}))}^{2}.
\end{multlined}
\end{equation}
Inserting the inductive hypothesis into \eqref{eqn:contraction-properties_complex-contraction_multiple-iteration_inductive-step} completes the proof.
\end{proof}

All our results so far have been general and do not depend on the specifics of the tree recursions.
At this time, however, we specialize to the ones constructed in \S\ref{sec:tree-recursions}.

\begin{corollary}
\label{cor:contraction-properties_complex-contraction_direct-estimate}
Let $\bl$ be an activity function taking values in $\mathcal{N}_{e^{-\lsconst}\delta_{1}}([0,e^{-\lsconst}\lambda])$ and $\bl^{m}$ a modulation of $\bl$.
Let $\vec{\gamma}$ be the depth-$k$ modulating function constructed in Theorem \ref{thm:tree-recursions_density-correspondence}.
Suppose that $\vec{\tau}_{1}$ and $\vec{\tau}_{2}$ are depth-$k$ boundary conditions taking values in $\psi^{-1}(\mathcal{N}_{\delta_{2}}[\sqrt{\delta},\sqrt{\delta+\tz^{2}}])$.
Then,
\begin{equation}
\abs{\psi(\pi_{\bl^{m},\vec{\tau}_{1},\vec{\gamma}}(v_{0}))
-\psi(\pi_{\bl^{m},\vec{\tau}_{2},\vec{\gamma}}(v_{0}))}
\le e^{-k \epsilon}
\norm{\psi(\vec{\tau}_{1})-\psi(\vec{\tau}_{2})}_{\infty}.
\end{equation}
\end{corollary}

\begin{proof}
Recall from the proof of Claim \ref{clm:tree-recursions_density-correspondence_interpolation} that, for all $1\le j\le k$ and $v_{0},\dots,v_{j-1}\in\R^{d}$, $\vec{\gamma}(v_{0},\dots,v_{j-1},\cdot)\bl^{m}(\cdot)$ is a modulation of $\bl$ and thus takes values in $\mathcal{N}_{\delta_{1}}([0,\lambda])$ by local stability (observe that the prefactor in \eqref{eqn:tree-recursions_modulation} is in $[0,e^{\lsconst}]$).
Hence, applying Proposition \ref{prop:contraction-properties_complex-contraction_multiple-iteration} to $\pi_{\bl^{m},\vec{\tau}_{i},\vec{\gamma}}$ and setting $j=1$ (in the resulting version of \eqref{eqn:contraction-properties_complex-contraction_multiple-iteration}), we have that
\begin{equation}
\begin{multlined}
\abs{\psi(\pi_{\bl^{m},\vec{\tau}_{1},\vec{\gamma}}(v_{0}))
-\psi(\pi_{\bl^{m},\vec{\tau}_{2},\vec{\gamma}}(v_{0}))}^{2}
\le (e^{\epsilon}M(\tilde{z}))^{k}\int_{(\R^{d})^{k}}\dd{v}_{1}\dots\dd{v}_{k}
\\
\times \prod_{\ell=1}^{k}\left[\abs{\bl^{m}(v_{\ell-1})}\vec{\gamma}(v_{0},\dots,v_{\ell-1},v_{\ell-1})\vphantom{\abs{1-e^{-\phi(v_{\ell-1}-v_{\ell})}}}\abs{1-e^{-\phi(v_{\ell-1}-v_{\ell})}}\right]
\abs{\psi(\vec{\tau}_{1}(v_{0},\dots,v_{k}))-\psi(\vec{\tau}_{2}(v_{0},\dots,v_{k}))}^{2}.
\end{multlined}
\end{equation}
Bounding $\abs{\bl^{m}(v_{\ell-1})}\le e^{\lsconst}\abs{\bl(v_{\ell-1})}\le e^{\epsilon}\lambda$ using Property \ref{itm:contraction-properties_complex-contraction_choices-of-constants_activity-modulus}, we have that
\begin{equation}
\abs{\psi(\pi_{\bl,\vec{\tau}_{1},\vec{\gamma}}(v_{0}))
-\psi(\pi_{\bl,\vec{\tau}_{2},\vec{\gamma}}(v_{0}))}^{2}
\le (e^{2\epsilon}M(\tilde{z})\lambda)^{k}V_{k,v_{0}}(\vec{\gamma})
\norm{\psi(\vec{\tau}_{1})-\psi(\vec{\tau}_{2})}_{\infty}^{2}.
\end{equation}
Recalling that $\lambda\le e^{-5\epsilon}\tl$ and bounding $V_{k,v_{0}}(\vec{\gamma})\le V_{k,v_{0}}(\vec{\gamma}_{c})$ as in \eqref{eqn:contraction-properties_real-contraction_Vk-domination} (and using translation invariance), we get that
\begin{equation}
\abs{\psi(\pi_{\bl,\vec{\tau}_{1},\vec{\gamma}}(v_{0}))
-\psi(\pi_{\bl,\vec{\tau}_{2},\vec{\gamma}}(v_{0}))}^{2}
\le (e^{-3\epsilon}M(\tilde{z})\tl)^{k}V_{k,v_{0}}(\vec{\gamma}_{c})
\norm{\psi(\vec{\tau}_{1})-\psi(\vec{\tau}_{2})}_{\infty}^{2}.
\end{equation}
Finally, bounding $V_{k,v_{0}}(\vec{\gamma}_{c})\le (e^{\epsilon}\pwcc)^{k}$ and noting that $\tl M(\tz)\pwcc\le 1$ (see \eqref{eqn:contraction-properties_optimization-problem}) complete the proof.
\end{proof}

At last, we package Corollary \ref{cor:contraction-properties_complex-contraction_direct-estimate} in a form suitable for application in \S\ref{sec:zero-freeness} using the following uniform continuity property adapted from \cite[Lemma 41]{michelen2023analyticity}.

\begin{lemma}
\label{lem:contraction-properties_complex-contraction_uniformly-continuous-tree-recursion}
With the same constants as chosen in \S\ref{sec:contraction-properties_complex-contraction_choices-of-constants}, for all $\epsilon'>0$, there exists $\delta'\in(0,\delta_{1}]$ such that the following holds for all activity functions $\bl,\bl'$ taking values in $\mathcal{N}_{e^{-\lsconst}\delta'}([0,e^{-\lsconst}\lambda])$ with $\norm{\bl-\bl'}_{\infty}\le e^{-\lsconst}\delta'$.
For all modulations $\bl^{m}$ of $\bl$, depth-$k$ boundary conditions $\vec{\tau}$ taking values in $\psi^{-1}(\mathcal{N}_{\delta_{2}}([\sqrt{\delta},\sqrt{\delta+\tz^{2}}]))$, depth-$k$ modulating functions $\vec{\gamma}$ as in \eqref{eqn:tree-recursions_modulating-function-choice}, and $v_{0}\in\R^{d}$, 
\begin{equation}
\abs{\psi(\pi_{\bl^{m},\vec{\tau},\vec{\gamma}}(v_{0}))-\psi(\pi_{\bl^{\prime,m},\vec{\tau},\vec{\gamma}}(v_{0}))}\le\epsilon',
\end{equation}
where $\bl^{\prime,m}$ denotes the modulation of $\bl'$ corresponding to $\bl^{m}$ (see \eqref{eqn:tree-recursions_modulation}):
\begin{equation}
\bl^{\prime,m}(\cdot)=\alpha\indicator{A}(\cdot)\prod_{i=1}^{n}\left[1+(e^{-\phi(x_{i}-\cdot)}-1)\indicator{B_{t_{i}}(x_{i})}(\cdot)\right]\bl'(\cdot).
\end{equation}
\end{lemma}

As the proof of Lemma \ref{lem:contraction-properties_complex-contraction_uniformly-continuous-tree-recursion} is technical and not particularly illuminating, we omit it from here and refer the reader to the discussion before \cite[Lemma 41]{michelen2023analyticity} for the main ideas.
We note that the proof is similar to that of Lemma \ref{lem:contraction-properties_complex-contraction_deviation-extent}: we only need to select an appropriate counterpart of the constant $B$ used there and then proceed by induction, exploiting the fact that we are working in a region of analyticity of the potential function $\psi$ thanks to Property \ref{itm:contraction-properties_complex-contraction_choices-of-constants_function-well-definedness}.

\begin{corollary}
\label{cor:contraction-properties_complex-contraction_final-result}
For all $\lambda\in[0,\tl)$, there exist $k\ge 1$, $\epsilon>0$, and complex neighborhoods $U_{1},U_{2}$ of $[0,\tz^{2}]$ with $\overline{U}_{1}\subset U_{2}$ such that the following holds.
Let $\bl$ be an activity function taking values in $\mathcal{N}_{\epsilon}([0,e^{-\lsconst}\lambda])$ and $\bl^{m}$ a modulation of $\bl$.
If the depth-$k$ boundary condition $\vec{\tau}$ in \eqref{eqn:tree-recursions_boundary-condition-choice} takes values in $\overline{U}_{2}$, then, with the depth-$k$ modulating function $\vec{\gamma}$ as in \eqref{eqn:tree-recursions_modulating-function-choice}, the depth-$k$ tree recursion $\pi_{\bl^{m},\vec{\tau},\vec{\gamma}}(v_{0})$ takes values in $U_{1}$ for all $v_{0}\in\R^{d}$.
\end{corollary}

\begin{proof}
With the same constants as chosen in \S\ref{sec:contraction-properties_complex-contraction_choices-of-constants}, we invoke Lemma \ref{lem:contraction-properties_complex-contraction_uniformly-continuous-tree-recursion} with $\epsilon'=(1-e^{-k\epsilon/2})\delta_{2}$ to get the corresponding $\delta'\in(0,\delta_{1}]$ such that its conclusion holds.
Let $\bl$ be an activity function taking values in $\mathcal{N}_{e^{-\lsconst}\delta'}([0,e^{-\lsconst}\lambda])$, $\bl^{m}$ a modulation of $\bl$, $\vec{\gamma}$ the depth-$k$ modulating function as in \eqref{eqn:tree-recursions_modulating-function-choice}, and $\vec{\tau}$ a depth-$k$ boundary condition taking values in $\psi^{-1}(\mathcal{N}_{\delta_{2}}([\sqrt{\delta},\sqrt{\delta+\tz^{2}}]))$.
As in the proof of Lemma \ref{lem:contraction-properties_complex-contraction_deviation-extent}, we introduce the projection maps $\pi_{1}(z)=\Re z$ for $z\in\C$, and
\begin{align}
\pi_{2}^{a}(x){}&:=x\indicator{[0,e^{-\lsconst}\lambda]}(x)+e^{-\lsconst}\lambda\indicator{(e^{-\lsconst}\lambda,\infty)}(x), \\
\pi_{2}^{b}(x){}&:=\sqrt{\delta}\indicator{(-\infty,\sqrt{\delta})}(x)+x\indicator{[\sqrt{\delta},\sqrt{\delta+\tz^{2}}]}(x)+\sqrt{\delta+\tz^{2}}\indicator{(\sqrt{\delta+\tz^{2}},\infty)}(x).
\end{align}
Since all these projection maps are continuous, the compositions $\bl':=\pi_{2}^{a}\circ\pi_{1}\circ\bl$ and $\vec{\tau}':=\psi^{-1}\circ\pi_{2}^{b}\circ\pi_{1}\circ\psi\circ\vec{\tau}$ are respectively an activity function and a depth-$k$ boundary condition such that $\norm{\bl-\bl'}_{\infty}\le e^{-\lsconst}\delta'$ and $\norm{\psi(\vec{\tau})-\psi(\vec{\tau}')}_{\infty}\le\delta_{2}$.
Moreover, $\bl'$ takes values in the real interval $[0,e^{-\lsconst}\lambda]$ and $\vec{\tau}'$ in $\psi^{-1}([\sqrt{\delta},\sqrt{\delta+\tz^{2}}])=[0,\tz^{2}]$.
Let $\bl^{\prime,m}$ be the modulation of $\bl'$ corresponding to $\bl^{m}$.
On the one hand, we have that
\begin{equation}
\label{eqn:contraction-properties_complex-contraction_final-result_change-boundary-condition}
\abs{\psi(\pi_{\bl^{m},\vec{\tau},\vec{\gamma}}(v_{0}))
-\psi(\pi_{\bl^{m},\vec{\tau}',\vec{\gamma}}(v_{0}))}
\le e^{-k \epsilon}
\norm{\psi(\vec{\tau})-\psi(\vec{\tau}')}_{\infty}
\le e^{-k \epsilon}\delta_{2}.
\end{equation}
by Corollary \ref{cor:contraction-properties_complex-contraction_direct-estimate}.
On the other hand, we have that
\begin{equation}
\label{eqn:contraction-properties_complex-contraction_final-result_change-activity-function}
\abs{\psi(\pi_{\bl^{m},\vec{\tau}',\vec{\gamma}}(v_{0}))-\psi(\pi_{\bl^{\prime,m},\vec{\tau}',\vec{\gamma}}(v_{0}))}\le(1-e^{-k\epsilon/2})\delta_{2}
\end{equation}
by Lemma \ref{lem:contraction-properties_complex-contraction_uniformly-continuous-tree-recursion}.
Combining \eqref{eqn:contraction-properties_complex-contraction_final-result_change-boundary-condition} and \eqref{eqn:contraction-properties_complex-contraction_final-result_change-activity-function} shows that
\begin{equation}
\abs{\psi(\pi_{\bl^{m},\vec{\tau},\vec{\gamma}}(v_{0}))
-\psi(\pi_{\bl^{\prime,m},\vec{\tau}',\vec{\gamma}}(v_{0}))}
\le(1-e^{-k\epsilon/2}+e^{-k\epsilon})\delta_{2}=:\delta_{3}<\delta_{2}.
\end{equation}
We note that it follows from \eqref{eqn:contraction-properties_complex-contraction_self-map-condition} that $\psi(\pi_{\bl^{\prime,m},\vec{\tau}',\vec{\gamma}}(v_{0}))\in[0,\tz^{2}]$.
Thus, we have proven
that $\psi(\pi_{\bl^{m},\vec{\tau},\vec{\gamma}}(v_{0}))\in\mathcal{N}_{\delta_{3}}([0,\tz^{2}])$.
The corollary follows by taking $k$ as in \S\ref{sec:contraction-properties_complex-contraction_choices-of-constants}, $\epsilon=e^{-\lsconst}\delta'$, $U_{1}=\psi^{-1}(\mathcal{N}_{\delta_{3}}([0,\tz^{2}]))$, and $U_{2}=\psi^{-1}(\mathcal{N}_{\delta_{2}}([0,\tz^{2}]))$, after noting that $\psi^{-1}$ is a conformal map (see also the proof of Lemma \ref{lem:contraction-properties_complex-contraction_convex-image}).
\end{proof}

\section{Zero-freeness}
\label{sec:zero-freeness}

In this section, we use the contraction properties of finite-depth tree recursions (see Corollary \ref{cor:contraction-properties_complex-contraction_final-result}) to prove that, for all $\lambda\in[0,e^{-\lsconst}\tl)$ and complex-valued activity functions $\bl$ taking values in a neighborhood of $[0,\lambda]$, the finite-volume pressure $\frac{1}{\abs{\Lambda}}\log Z_{\Lambda}(\bl)$ is uniformly bounded in modulus.
This implies, in particular, that the partition function $Z_{\Lambda}(\lambda)$ is nonzero for activities $\lambda$ in this complex domain, for all bounded regions $\Lambda\subset\R^{d}$.

To this end, we will heavily exploit the uniform continuity of the partition function with respect to its activity function:

\begin{lemma}
\label{lem:zero-freeness_Z-uniform-continuity}
For any bounded, Lebesgue measurable region $\Lambda\subset\R^{d}$ and $M>0$, the map $\bl\mapsto Z_{\Lambda}(\bl)$, defined on the set of activity functions $\bl$ in $\Lambda$ with $\abs{\bl(x)}\le M$ for all $x\in\Lambda$, is uniformly continuous with respect to the $\infty$-norm.
\end{lemma}

\begin{proof}
The proof requires only minor modification from that of \cite[Lemma 9]{michelen2023analyticity} thanks to the stability of the pair potential.
We present it here for completeness.
Let $\bl_{1}$, $\bl_{2}$ be activity functions on $\Lambda$ satisfying $\norm{\bl_{1}}_{\infty},\norm{\bl_{2}}_{\infty}\le M$ with $M\ge 1$ and $\norm{\bl_{1}-\bl_{2}}_{\infty}\le\delta$.
Then,
\begin{equation}
\abs{Z_{\Lambda}(\bl_{1})-Z_{\Lambda}(\bl_{2})}
\le\sum_{n=1}^{\infty}\frac{1}{n!}\int_{\Lambda^{n}}\dd{x}_{1}\dots\dd{x}_{n}\abs{\prod_{i=1}^{n}\bl_{1}(x_{i})-\prod_{i=1}^{n}\bl_{2}(x_{i})}\cdot e^{-U(x_{1},\dots,x_{n})}.
\end{equation}
By \cite[Lemma 10]{michelen2023analyticity}, we bound
\begin{equation}
\abs{\prod_{i=1}^{n}\bl_{1}(x_{i})-\prod_{i=1}^{n}\bl_{2}(x_{i})}
\le nM^{n-1}\delta.
\end{equation}
Moreover, we bound $e^{-U(x_{1},\dots,x_{n})}\le e^{\lsconst n/2}$ using stability (see \eqref{eqn:introduction_stability}).
Thus,
\begin{equation}
\abs{Z_{\Lambda}(\bl_{1})-Z_{\Lambda}(\bl_{2})}
\le\sum_{n=1}^{\infty}\frac{1}{n!}\abs{\Lambda}^{n}nM^{n-1}\delta e^{\lsconst n/2}
=e^{\lsconst/2}\abs{\Lambda}e^{e^{\lsconst/2}M\abs{\Lambda}}\delta,
\end{equation}
which proves uniform continuity.
\end{proof}

\begin{theorem}
\label{thm:zero-freeness}
For each $\lambda\in[0,e^{-\lsconst}\tilde{\lambda})$, there exist constants $\epsilon,c>0$ such that the following holds for all $\Lambda\subset\R^{d}$.
Let $\bl:\Lambda\rightarrow\C$ be an activity function taking values in $\mathcal{N}_{\epsilon}([0,\lambda])$.
Then, $Z_{\Lambda}(\bl)\ne 0$, and
\begin{equation}
\abs{\log Z_{\Lambda}(\bl)}\le c\abs{\Lambda}.
\end{equation}
\end{theorem}

\begin{proof}
Let $\lambda\in[0,e^{-\lsconst}\tilde{\lambda})$.
Applying Corollary \ref{cor:contraction-properties_complex-contraction_final-result} with to the activity parameter $e^{\lsconst}\lambda$, we obtain $k\ge 1$, $\epsilon>0$, and complex neighborhoods $U_{1},U_{2}$ of $[0,\tz^{2}]$ such that its conclusion holds.
Let $\bl:\Lambda\rightarrow\C$ be an activity function which takes values in $\mathcal{N}_{\epsilon}([0,\lambda])$.

\subparagraph{Interpolation}

Let us define $\bl_{s}(\cdot):=s\bl(\cdot)$ for $s\in[0,1)$ and set $\bl_{1}(\cdot):=\bl(\cdot)$ so that $\bl_{s}$ interpolates between zero activity and $\bl$.
We denote by $\mathcal{M}_{\ast}$ the set of $\bl_{s}$ which are totally zero-free (recall Definition \ref{def:tree-recursions_totally-zero-free}) and have the property that the one-point density of every modulation of $\bl_{s}$ takes values in $\overline{U}_{2}$:
\begin{equation}
\begin{multlined}
\mathcal{M}_{\ast}:=\left\{\bl_{s}\mid s\in[0,1],\ \bl_{s}\text{ is totally zero-free,}\right.
\\
\left.\text{ and }\rho_{\bl_{s}^{m}}(v_{0})\in\overline{U}_{2}\text{ for all }v_{0}\in\R^{d}\text{ and modulations }\bl_{s}^{m}\text{ of }\bl_{s}\right\}.
\end{multlined}
\end{equation}
Following the convention earlier, we denote by $\bl_{s}^{m}$ a generic modulation of $\bl_{s}$.

We will prove that $\bl=\bl_{1}\in\mathcal{M}_{\ast}$, that is, $\bl$ is totally zero-free, and
\begin{equation}
\label{eqn:zero-freeness_property}
\rho_{\bl^{m}}(v_{0})\in\overline{U}_{2}\text{ for all }v_{0}\in\R^{d}\text{ and modulations }\bl^{m}\text{ of }\bl.
\end{equation}
Note that this immediately implies the theorem: $Z_{\Lambda}(\bl)\ne 0$ because $\bl$ is totally zero-free, and, setting $c:=\max\set{\abs{z}\mid z\in\overline{U}_{2}}$ and applying Lemma \ref{lem:one-point-densities_log-Z}, we have that
\begin{equation}
\label{eqn:zero-freeness_desired-bound}
\abs{\log Z_{\Lambda}(\bl)}\le\int_{\Lambda}\dd{x}\abs{\rho_{\hat{\bl}_{x}}(x)}\le c\abs{\Lambda}.
\end{equation} 

Suppose by contradiction that $\bl_{1}\not\in\mathcal{M}_{\ast}$.
We observe the following elementary properties of the interpolated activity functions:
\begin{enumerate}
\item $\bl_{1}\not\in\mathcal{M}_{\ast}$, which is our assumption.
\item $\bl_{0}\in\mathcal{M}_{\ast}$, since the only modulation of $\bl_{0}$ is the identically zero activity function, whose one-point density is identically zero by \eqref{eqn:one-point-densities_def}.
\item The membership of $\bl_{s}$ in $\mathcal{M}_{\ast}$ is decreasing in $s$: if $0<s'<s\le 1$ and $\bl_{s}\in\mathcal{M}_{\ast}$, then $\bl_{s'}\in\mathcal{M}_{\ast}$.
This is because every modulation of $\bl_{s'}$ is a modulation of $\bl_{s}$ (recall \eqref{eqn:tree-recursions_modulation}):
\begin{equation}
\begin{multlined}
\alpha\indicator{A}(\cdot)\prod_{i=1}^{n}\left[1+(e^{-\phi(x_{i}-\cdot)}-1)\indicator{B_{t_{i}}(x_{i})}(\cdot)\right]\bl_{s'}(\cdot)
\\
=(s'\alpha/s)\indicator{A}(\cdot)\prod_{i=1}^{n}\left[1+(e^{-\phi(x_{i}-\cdot)}-1)\indicator{B_{t_{i}}(x_{i})}(\cdot)\right]\bl_{s}(\cdot).
\end{multlined}
\end{equation}
\end{enumerate} 
Hence, it is natural to examine where the membership of $\bl_{s}$ in $\mathcal{M}_{\ast}$ ends: let
\begin{equation}
\label{eqn:zero-freeness_supremum}
s_{\ast}:=\sup\set{s\in[0,1]\mid\bl_{s}\in\mathcal{M}_{\ast}},
\end{equation}
and introduce the shorthand $\bl_{*}:=\bl_{s_{\ast}}$.

\subparagraph{Properties of $s_{\ast}$}

Now, we show that $s_{\ast}\in(0,1)$.

First, we prove that $s_{\ast}>0$.
We rely on the following technical result which we will reuse in a later part of the proof.
Essentially, it gives a condition under which we can extend the interval of $s$ for which $\bl_{s}$ is totally zero-free, while maintaining control on the growth (in modulus) of the one-point densities of \emph{all} modulations of $\bl_{s}$ over the extended interval.

\begin{claim}
\label{clm:zero-freeness_enlargement}
Suppose that $\bl_{s}$ is totally zero-free for some $s\in[0,1)$ and that there exists a constant $c>0$ such that $\abs{\log Z(\bl_{s}^{m})}\le c\abs{\Lambda}$ for all modulations $\bl_{s}^{m}$ of $\bl_{s}$.
Then, for all $\epsilon>0$, there exists $s'\in(s,1]$ such that:
\begin{enumerate}
\item $\bl_{s'}$ is totally zero-free, and \label{itm:zero-freeness_enlargement_1}
\item for all modulations $\bl_{s'}^{m}$ of $\bl_{s'}$, there exists a corresponding modulation $\bl_{s}^{m}$ of $\bl_{s}$ such that $\norm{\rho_{\bl_{s'}^{m}}-\rho_{\bl_{s}^{m}}}_{\infty}\le \epsilon$. \label{itm:zero-freeness_enlargement_2}
\end{enumerate}
\end{claim}

Claim \ref{clm:zero-freeness_enlargement} is a consequence of the uniform continuity established in Lemma \ref{lem:zero-freeness_Z-uniform-continuity}; we defer its proof to the end of the section.
From the claim, it immediately follows that $s_{\ast}>0$.
Indeed, since $\bl_{0}\equiv 0$ clearly satisfies its assumptions, it allows us to find $s'\in(0,1]$ such that $\bl_{s'}$ is totally zero-free and, for all modulations $\bl_{s'}^{m}$ of $\bl_{s'}$, $\norm{\rho_{\bl_{s'}^{m}}}_{\infty}\le\frac{1}{2}d(0,\partial U_{2})$ (recall that $\rho_{0}\equiv 0$ as observed before), which implies that $s'\in\mathcal{M}_{\ast}$.

Next, we prove that $s_{\ast}<1$.
To this end, we use another consequence of uniform continuity, whose proof we again defer to the end of the section.
\begin{claim}
\label{clm:zero-freeness_criticality}
$\bl_{\ast}\in\mathcal{M}_{\ast}$.
\end{claim}
By assumption, $\bl_{1}\not\in\mathcal{M}_{\ast}$, so $s_{\ast}\ne 1$ by Claim \ref{clm:zero-freeness_criticality}, which forces $s_{\ast}<1$.

\subparagraph{Recursion}

Next, we prove the stronger statement that the one-point density of every modulation of $\bl_{\ast}$ actually takes values in $U_{1}$ (as opposed to just $\overline{U}_{2}$):
\begin{equation}
\rho_{\bl^{m}_{\ast}}(v_{0})\in U_{1}\text{ for all }v_{0}\in\R^{d}\text{ and modulations }\bl^{m}_{\ast}\text{ of }\bl_{\ast}.
\end{equation}
Let $\bl^{m}_{\ast}$ be a modulation of $\bl_{\ast}$.
By Theorem \ref{thm:tree-recursions_density-correspondence}, we construct a depth-$k$ tree recursion which computes the one-point density $\rho_{\bl^{m}_{\ast}}(v_{0})$:
\begin{equation}
\rho_{\bl^{m}_{\ast}}(v_{0})=\pi_{\bl^{m}_{\ast},\vec{\tau},\vec{\gamma}}(v_{0}),
\end{equation}
where we recall from \eqref{eqn:tree-recursions_boundary-condition-choice} that the boundary condition $\vec{\tau}$ is given by 
\begin{equation}
\vec{\tau}(v_{0},\dots,v_{k}):=
\begin{cases}
\rho_{\bl^{m}_{\ast,v_{0}\rightarrow\dots\rightarrow v_{k}}}(v_{k}) & \text{if } (v_{0},\dots,v_{k})\in G^{(k+1)} \\
0 & \text{else}
\end{cases}.
\end{equation}
By the same proof as of Claim \ref{clm:tree-recursions_density-correspondence_interpolation}, for all $(v_{0},\dots,v_{k})\in G^{(k+1)}$, the activity function $\bl^{m}_{\ast,v_{0}\rightarrow\dots\rightarrow v_{k}}$ is a modulation of $\bl^{m}_{\ast}$, so $\rho_{\bl^{m}_{\ast,v_{0}\rightarrow\dots\rightarrow v_{k}}}(v_{k})\in\overline{U}_{2}$ by Claim \ref{clm:zero-freeness_criticality}.
Since $0\in\overline{U}_{2}$ also, we conclude that the boundary condition $\vec{\tau}$ takes values in $\overline{U}_{2}$.
By Corollary \ref{cor:contraction-properties_complex-contraction_final-result}, this implies that $\rho_{\bl_{\ast}^{m}}(v)\in U_{1}$, as required.

\subparagraph{Contradiction}

Since $d(\overline{U}_{1},U_{2}^{c})>0$, we can apply Claim \ref{clm:zero-freeness_enlargement} again to find $s'\in(s_{\ast},1]$ such that $\bl_{s'}$ is totally zero-free and, for all modulations $\bl_{s'}^{m}$ of $\bl_{s'}$, there exists a corresponding modulation $\bl_{s}^{m}$ of $\bl_{s}$ such that $\norm{\rho_{\bl_{s'}^{m}}-\rho_{\bl_{s}^{m}}}_{\infty}\le\frac{1}{2}d(\overline{U}_{1},U_{2}^{c})$.
However, this implies that $\rho_{\bl_{s'}^{m}}$ takes values in $\overline{U}_{2}$ for all modulations $\bl_{s'}^{m}$ of $\bl_{s'}$, and hence that $\bl_{s'}\in\mathcal{M}_{\ast}$, which contradicts the definition \eqref{eqn:zero-freeness_supremum} of $s_{\ast}$.
\end{proof}

Our main theorem (see Theorem \ref{thm:analyticity}) on the analyticity of the infinite-volume pressure follows immediately from the zero-freeness of the finite-volume partition functions (see Theorem \ref{thm:zero-freeness}) by an application of Vitali's convergence theorem.

\begin{theorem}[Vitali's convergence theorem, {\cite[Theorem 20]{michelen2023analyticity}}]
\label{thm:Vitali}
Let $\Omega\subset\C$ be a domain and $f_{n}$ a sequence of analytic functions on $\Omega$ such that $\abs{f_{n}(z)}\le M$ for some $M>0$ for all $n$ and $z\in\Omega$.
If there exists a sequence of distinct points $z_{m}\in\Omega$ with $z_{m}\rightarrow z_{\infty}\in\Omega$ such that $\lim_{n\rightarrow\infty}f_{n}(z_{m})$ exists for all $m$, then $f_{n}$ converges uniformly on compact subsets of $\Omega$ to an analytic function $f$ on $\Omega$.
\end{theorem}

\begin{proof}[Proof of Theorem \ref{thm:analyticity}]
Let $\lambda_{0}\in[0,e^{-\lsconst}\tilde{\lambda})$, and consider the complex domain $\mathcal{N}_{\epsilon}([0,\lambda_{0}])$ given by Theorem \ref{thm:zero-freeness}.
By Theorem \ref{thm:zero-freeness}, for any bounded, Lebesgue measurable region $\Lambda\subset\R^{d}$, the finite-volume pressure $p_{\Lambda}(\lambda):=\frac{1}{\abs{\Lambda}}\log Z_{\Lambda}(\lambda)$ is an analytic function of $\lambda$ on $\mathcal{N}_{\epsilon}([0,\lambda_{0}])$, and on which it is uniformly bounded by a constant $c>0$.
It is well-known that, for $\lambda\ge 0$, the limit $\lim_{\Lambda\uparrow\R^{d}}p_{\Lambda}(\lambda)$ exists and is equal to $p_{\phi}(\lambda)$.
Therefore, by Theorem \ref{thm:Vitali}, $p_{\Lambda}(\lambda)$ converges uniformly on compact subsets of $\mathcal{N}_{\epsilon}([0,\lambda_{0}])$ to an analytic function.
\end{proof}

At last, we prove the two claims used in the proof of Theorem \ref{thm:zero-freeness}.

\begin{proof}[Proof of Claim \ref{clm:zero-freeness_enlargement}]
We first establish Property \ref{itm:zero-freeness_enlargement_1}.
Let $s'\in(s,1]$ and $\bl_{s'}^{m}$ a modulation of $\bl_{s'}$.
Notice that $s\bl_{s'}^{m}/s'$ is a modulation of $\bl_{s}$.
Using the elementary inequality $\abs{\log\abs{z}}\le\abs{\log z}$ which holds for all $z\in\C\setminus\set{0}$, we get that 
\begin{equation}
\label{eqn:zero-freeness_enlargement_log-Z}
\abs{\log\abs{Z(s\bl_{s'}^{m}/s')}}\le\abs{\log Z(s\bl_{s'}^{m}/s')}\le c\abs{\Lambda}.
\end{equation}
At the same time, by Lemmas \ref{lem:one-point-densities_uniform-boundedness} and \ref{lem:zero-freeness_Z-uniform-continuity}, there exists $\delta_{1}>0$ such that, whenever two modulations $\bl^{m,1},\bl^{m,2}$ of $\bl$ satisfies $\norm{\bl^{m,1}-\bl^{m,2}}_{\infty}\le\delta_{1}$, we have that 
\begin{equation}
\label{eqn:zero-freeness_enlargement_Z-uniform-continuity}
\abs{Z(\bl^{m,1})-Z(\bl^{m,2})}\le\frac{1}{2e^{c\abs{\Lambda}}}.
\end{equation}
Since 
\begin{equation}
\label{eqn:zero-freeness_enlargement_activity-difference}
\norm{\bl_{s'}^{m}-s\bl_{s'}^{m}/s'}_{\infty}=(1-s/s')\norm{\bl_{s'}^{m}}_{\infty}\le(s'-s)e^{\lsconst}\norm{\bl}_{\infty},
\end{equation}
we see that, if we take $s'$ so close to $s$ that
\begin{equation}
\label{eqn:zero-freeness_enlargement_zero-free-condition}
(s'-s)e^{\lsconst}\norm{\bl}_{\infty}\le\delta_{1},
\end{equation}
then, by \eqref{eqn:zero-freeness_enlargement_log-Z} and \eqref{eqn:zero-freeness_enlargement_Z-uniform-continuity},
\begin{equation}
\label{eqn:zero-freeness_enlargement_enlarged-modulation-lower-bound}
\abs{Z(\bl_{s'}^{m})}\ge\abs{Z(s\bl_{s'}^{m}/s')}-\abs{Z(\bl_{s'}^{m})-Z(s\bl_{s'}^{m}/s')}\ge\frac{1}{2e^{c\abs{\Lambda}}}.
\end{equation}
Therefore, as long as $s'\in(s,1]$ satisfies the condition \eqref{eqn:zero-freeness_enlargement_zero-free-condition}, $\bl_{s'}$ is totally zero-free.

We now move on to Property \ref{itm:zero-freeness_enlargement_2}.
Suppose that $s'\in(s,1]$ satisfies \eqref{eqn:zero-freeness_enlargement_zero-free-condition}.
Let $\bl_{s'}^{m}$ be a modulation of $\bl_{s'}$, and consider again the modulation $s\bl_{s'}^{m}/s'$ of $\bl_{s}$.
By total zero-freeness, both $\bl_{s'}^{m}$ and $s\bl_{s'}^{m}/s'$ have well-defined one-point densities.
Fixing $v\in\R^{d}$ and recalling \eqref{eqn:one-point-densities_def}, we bound
\begin{equation}
\label{eqn:zero-freeness_enlargement_density-difference}
\begin{multlined}
\abs{\rho_{\bl_{s'}^{m}}(v)-\rho_{s\bl_{s'}^{m}/s'}(v)}
\le\frac{\abs{\bl_{s'}^{m}(v)}\abs{Z(\bl_{s'}^{m}e^{-\phi(v-\cdot)})}}{\abs{Z(\bl_{s'}^{m})}\abs{Z(s\bl_{s'}^{m}/s')}}\abs{Z(s\bl_{s'}^{m}/s')-Z(\bl_{s'}^{m})}
\\
+\frac{\abs{\bl_{s'}^{m}(v)}}{\abs{Z(s\bl_{s'}^{m}/s')}}\abs{Z(\bl_{s'}^{m}e^{-\phi(v-\cdot)})-Z(s\bl_{s'}^{m}e^{-\phi(v-\cdot)}/s')}
\\
+\frac{\abs{Z(s\bl_{s'}^{m}e^{-\phi(v-\cdot)}/s')}}{\abs{Z(s\bl_{s'}^{m}/s')}}\abs{\bl_{s'}^{m}(v)-s\bl_{s'}^{m}(v)/s'}
\end{multlined}
\end{equation}
(that is, we use the triangle inequality to change the instances of $\bl_{s'}^{m}$ in the definition of $\rho_{\bl_{s'}^{m}}(v)$ one by one into $s\bl_{s'}^{m}/s'$).
Note that all the partition functions that appear in the denominators in \eqref{eqn:zero-freeness_enlargement_density-difference} involve modulations of $\bl_{s'}$, and so they satisfy the lower bound \eqref{eqn:zero-freeness_enlargement_enlarged-modulation-lower-bound}.
On the other hand, the activity functions and partition functions in the numerators are uniformly bounded above: taking $\abs{Z(\bl_{s'}^{m}e^{-\phi(v-\cdot)})}$ for example, since $\norm{\bl_{s'}^{m}e^{-\phi(v-\cdot)}}_{\infty}\le e^{2\lsconst}\norm{\bl}_{\infty}$, we have, as in the proof of Lemma \ref{lem:zero-freeness_Z-uniform-continuity}, that
\begin{equation}
\abs{Z(\bl_{s'}^{m}e^{-\phi(v-\cdot)})}
\le\sum_{n=0}^{\infty}\frac{1}{n!}\int_{\Lambda^{n}}\dd{x}_{1}\dots\dd{x}_{n}(e^{2\lsconst}\norm{\bl}_{\infty})^{n}e^{\lsconst n/2}
=\exp[e^{5\lsconst/2}\norm{\bl}_{\infty}\abs{\Lambda}].
\end{equation}
Thus, there exist constants $c_{1},c_{2},c_{3}>0$ such that 
\begin{equation}
\begin{multlined}
\abs{\rho_{\bl_{s'}^{m}}(v)-\rho_{s\bl_{s'}^{m}/s'}(v)}
\le c_{1}\abs{Z(s\bl_{s'}^{m}/s')-Z(\bl_{s'}^{m})}
\\
+c_{2}\abs{Z(\bl_{s'}^{m}e^{-\phi(v-\cdot)})-Z(s\bl_{s'}^{m}e^{-\phi(v-\cdot)}/s')}
+c_{3}\abs{\bl_{s'}^{m}(v)-s\bl_{s'}^{m}(v)/s'}.
\end{multlined}
\end{equation}
Let $\epsilon>0$.
By Lemma \ref{lem:zero-freeness_Z-uniform-continuity}, there exists $\delta_{2}>0$ such that, whenever two activity functions $\bl_{1},\bl_{2}$ bounded in modulus by $e^{2\lsconst}\norm{\bl}_{\infty}$ satisfy $\norm{\bl_{1}-\bl_{2}}_{\infty}\le\delta_{2}$, we have that 
\begin{equation}
\abs{Z(\bl_{1})-Z(\bl_{2})}\le\frac{\epsilon}{2(c_{1}+c_{2})}.
\end{equation}
Recalling \eqref{eqn:zero-freeness_enlargement_activity-difference}, by taking $s'$ so close to $s$ that
\begin{equation}
(s'-s)e^{\lsconst}\norm{\bl}_{\infty}\le\min\set{e^{-\lsconst}\delta_{2},\frac{\epsilon}{2c_{3}}},
\end{equation}
we find that $\norm{\rho_{\bl_{s'}^{m}}-\rho_{s\bl_{s'}^{m}/s'}}_{\infty}\le\epsilon$, as desired.
\end{proof}

\begin{proof}[Proof of Claim \ref{clm:zero-freeness_criticality}]
Let $\bl_{\ast}^{m}$ be a modulation of $\bl_{\ast}$, and fix $s\in[0,s_{\ast})$.
As observed earlier, $s\bl_{\ast}^{m}/s_{\ast}$ is a modulation of $\bl_{s}$, and $\bl_{s}\in\mathcal{M}_{\ast}$.
Thus, the bound \eqref{eqn:zero-freeness_desired-bound} holds:
\begin{equation}
\abs{\log Z_{\Lambda}(s\bl_{\ast}^{m}/s_{\ast})}\le c\abs{\Lambda}.
\end{equation}
Using the elementary inequality $\abs{\log\abs{z}}\le\abs{\log z}$ which holds for all $z\in\C\setminus\set{0}$, we get that 
\begin{equation}
\label{eqn:zero-freeness_Z-two-sided-bound}
e^{-c\abs{\Lambda}}\le\abs{Z_{\Lambda}(s\bl_{\ast}^{m}/s_{\ast})}\le e^{c\abs{\Lambda}}.
\end{equation}
By Lemma \ref{lem:zero-freeness_Z-uniform-continuity}, $Z_{\Lambda}(\bl_{\ast}^{m})=\lim_{s\rightarrow s_{\ast}}Z_{\Lambda}(s\bl_{\ast}^{m}/s_{\ast})$ again satisfies the bound \eqref{eqn:zero-freeness_Z-two-sided-bound}.
In particular, this shows that $Z(\bl_{\ast}^{m})\ne 0$.

To prove that $\rho_{\bl_{\ast}^{m}}$ takes values in $\overline{U}_{2}$, we fix $v\in\R^{d}$ and observe that, again by Lemma \ref{lem:zero-freeness_Z-uniform-continuity},
\begin{equation}
\rho_{\bl_{\ast}^{m}}(v)
=\bl_{\ast}^{m}(v)\frac{Z(\bl_{\ast}^{m}e^{-\phi(v-\cdot)})}{Z(\bl_{\ast}^{m})}
=\lim_{s\rightarrow s_{\ast}}s\bl_{\ast}^{m}(v)/s_{\ast}\frac{Z(s\bl_{\ast}^{m}e^{-\phi(v-\cdot)}/s_{\ast})}{Z(s\bl_{\ast}^{m}/s_{\ast})}
=\lim_{s\rightarrow s_{\ast}}\rho_{s\bl_{\ast}/s_{\ast}}(v).
\end{equation}
Since $\rho_{s\bl_{\ast}/s_{\ast}}$ takes values in $\overline{U}_{2}$ and $\overline{U}_{2}$ is closed, we conclude that $\rho_{\bl_{\ast}^{m}}(v)\in\overline{U}_{2}$, as desired.
\end{proof}

\subsection{Origin of the temperature dependence}
\label{sec:zero-freeness_origin}

With hindsight, we now reflect on the origin of the strong temperature dependence of our improvement factor of $e^{2-2W(e\aconst/\pwcc)}$ over the classical Penrose-Ruelle bound.
We focus on two milestones in this paper: Theorem \ref{thm:zero-freeness} where most of the action took place, and Proposition \ref{prop:contraction-properties_optimization-problem} which gives purely computational evidence for the inevitability of the temperature dependence.

We first note that although Theorem \ref{thm:zero-freeness} boils down to a statement \eqref{eqn:zero-freeness_property} about the values of certain one-point densities, yet nowhere in the proof have we ever actually computed any one-point density explicitly, except in the case of zero activity.
There is thus no a priori reason why, at nonzero activity, all the one-point densities should take values in the seemingly independently determined set $\overline{U}_{2}$.
Rather, the proof of Theorem \ref{thm:zero-freeness} makes it clear that this information only arises from examining the hierarchy of modulations of the fixed activity function $\bl$.
Essentially, the proof consists in arguing that the one-point density of every modulation $\bl^{m}$ of $\bl$ takes values in $\overline{U}_{2}$ because so is the case for all the $k$th descendants of $\bl^{m}$ in the hierarchy (the interpolation from zero activity makes all this rigorous).
In some sense, this involves taking information from infinitely far below the root of the hierarchy $\bl$ and passing it up the hierarchy $k$ steps at a time to ultimately establish control over the behavior of the root.

With this intuitive picture, we now examine the respective role of the self-map and contraction conditions in the argument.
The self-map condition \eqref{eqn:contraction-properties_real-contraction_self-map-condition} is first used in Lemma \ref{lem:contraction-properties_real-contraction_self-map-preservation} and Proposition \ref{prop:contraction-properties_real-contraction_multiple-iteration} to iterate the bound in Lemma \ref{lem:contraction-properties_real-contraction_single-iteration}.
In other words, it allows us to bridge between all the parents and children in the hierarchy at a uniform activity $\tl$, provided that there exists a corresponding $\tz>0$ such that the pair $(\tl,\tz)$ satisfies the self-map condition.
Indeed, in the case that $\aconst>0$, if we remove the self-map condition from consideration, then the exponential in \eqref{lem:contraction-properties_real-contraction_self-map-preservation_one-step} can in principle cause the information $\pi_{\bl,\vec{\tau},\vec{\gamma}}(v_{0},\dots,v_{j-1},w)$ from each child to ``grow'' a bit as it is passed upward to the parent node $(v_{0},\dots,v_{j-1})$.
Thus, it is possible that the information from the node $(v_{0},\dots,v_{j-1},w)$ will have ``overgrown'' so much by the time it reaches the root that nothing useful about the root may be extracted from there, let alone if we were to try to pass information upwards from infinitely far below as discussed earlier.
On the other hand, the role of the contraction condition \eqref{eqn:contraction-properties_real-contraction_contraction-condition} is much more concrete: with the nodes in the hierarchy organically connected by the self-map condition, the contraction condition allows us to control the behavior at any node using information passed upwards from its $k$th children. 
Therefore, both the self-map and contraction conditions appear indispensable in the current formulation of the argument.

This brings us to Proposition \ref{prop:contraction-properties_optimization-problem} where we studied the interplay of the two conditions and the implications for the activity threshold below which the argument applies.
The dependence of the contraction condition on the temperature is rather innocuous, since only $\pwcc$ has such a dependence and even that is implicit.
In stark contrast, the temperature dependence of the self-map condition is quite extreme.
As soon as we exit the purely repulsive setting studied by Michelen and Perkins into a regime where $\aconst>0$, the exponential factor in \eqref{eqn:contraction-properties_real-contraction_self-map-condition} becomes significant and grows rapidly as we increase $\beta$ (i.e., lowering the temperature).
As we computed in Proposition \ref{prop:contraction-properties_optimization-problem}, this unfortunate behavior of the self-map condition ends up taking a heavy toll on the performance of this approach in the low-temperature regime, as the overall improvement factor of $e^{2-2W(e\aconst/\pwcc)}$ tends to $1$ as $\beta\rightarrow\infty$.
On the other hand, this approach works much better in the high-temperature regime where $\aconst\ll 1$ so that the adverse effect of the self-map condition is almost negligible.

\section*{Acknowledgments}

The author thanks Will Perkins for his encouragement in the pursuit of this work and Ron Peled for instructive discussions.
The author also thanks an anonymous referee for useful feedback on an earlier version of this paper, and Ian Jauslin and Joel Lebowitz for helpful comments.

\bibliographystyle{plain}
\bibliography{bibliography.bib}

\end{document}